\documentclass[10pt,twocolumn,twoside]{IEEEtran} 
\IEEEoverridecommandlockouts 

\usepackage{graphicx}
\usepackage{color}
\usepackage{subfigure}
\usepackage{algorithm,algpseudocode}
\usepackage{amsmath, amssymb, amsthm, mathtools, accents}

\usepackage{amsmath,amssymb}
\usepackage[noadjust]{cite}
\usepackage{tikz}
	\usetikzlibrary{shapes,arrows}
	\tikzstyle{frame} = [draw, -latex]
	\tikzstyle{line} = [draw]
	\tikzstyle{line2} = [draw, dashdotted]
	\tikzstyle{line3} = [draw, dashed]
	\tikzstyle{line3UD} = [draw, dashed]
	\tikzstyle{place} = [circle, draw=black, fill=white, thick, inner sep=2pt, minimum size=1mm]
	\tikzstyle{place2} = [circle, draw=black, fill=black, thick, inner sep=2pt, minimum size=1mm]
	\tikzstyle{placeRed} = [circle, draw=red, fill=red, thick, inner sep=2pt, minimum size=1mm]
	\tikzstyle{vertex} = [circle, draw=black, fill=black, thick, inner sep=2pt, minimum size=1mm]
\usepackage{tkz-euclide}

\usetikzlibrary{shapes,arrows} 

\tikzstyle{decision} = [diamond, draw, fill=blue!20, 
    text width=4.5em, text badly centered, node distance=3cm, inner sep=0pt]
\tikzstyle{block1} = [rectangle, draw, text width=8em, text centered, minimum height=4em]
\tikzstyle{block2} = [rectangle, draw, text width=3em, text centered, minimum height=4em]
\tikzstyle{block3} = [rectangle, draw, text width=11em, text centered, minimum height=12em, dashed,black]
\tikzstyle{block4} = [rectangle, draw, text width=11em, text centered, minimum height=18em, dashed,black]
\tikzstyle{block5} = [rectangle, draw, text width=11em, text centered, minimum height=32em, dashed,black]
\tikzstyle{block6} = [rectangle, draw, text width=11em, text centered, minimum height=18.5em, dashed,black]
\tikzstyle{block7} = [rectangle, draw, text width=11em, text centered, minimum height=11.8em, dashed,black]
\tikzstyle{line01} = [draw, -latex']
\tikzstyle{line02} = [draw, latex'-latex']

\newcommand{\bbm}{\begin{bmatrix}}
\newcommand{\ebm}{\end{bmatrix}}

\usepackage{amsmath,bm}

\newtheorem{corollary}{\bfseries Corollary}
\newtheorem{definition}{\bfseries Definition }
\newtheorem{lemma}{\bfseries Lemma}

\newtheorem{remark}{\bfseries Remark}
\newtheorem{theorem}{\bfseries Theorem}
\newtheorem{assumption}{\bfseries Assumption}
\newtheorem{observation}{\bfseries Observation}
\newtheorem{claim}{\bfseries Claim}
\newtheorem{example}{\bfseries Example}

\definecolor{MG}{rgb}{0,0.45,0.08}

\graphicspath{{img/}}

\title{Cooperative opinion dynamics on multiple interdependent topics: Modeling and analysis}

\author{Hyo-Sung Ahn$^{1}$, Quoc Van Tran$^{1}$, Minh Hoang Trinh$^{1}$, Kevin L. Moore$^{2}$, Mengbin Ye$^{3}$, and Ji Liu$^{4}$

\thanks{${}^{1}$School of Mechanical Engineering, Gwangju Institute of Science and Technology (GIST), Gwangju, Korea. E-mails: {\tt\small hyosung@gist.ac.kr; tranvanquoc@gist.ac.kr; trinhhoangminh@gist.ac.kr}}
\thanks{${}^{2}$Department of Electrical Engineering, Colorado School of Mines, Golden, CO, USA. E-mails: {\tt\small kmoore@mines.edu}}
\thanks{${}^{3}$Systems and Control, Faculty of Science and Engineering, University of Groningen, Groningen 9747 AG, The Netherlands {\tt\small m.ye@rug.nl}}
\thanks{${}^{4}$Department of Electrical and Computer Engineering, 211 Light Engineering, Stony Brook University,
Stony Brook, NY, USA {\tt\small ji.liu@stonybrook.edu}}
}

\begin{document}
\maketitle
\thispagestyle{empty}
\pagestyle{empty}

\begin{abstract} 
To model the interdependent couplings of  multiple topics, we develop a set of rules for opinion updates of a group of agents. The rules are used to design or assign values to the elements of interdependent weighting matrices. The cooperative and anti-cooperative couplings are modeled in both the inverse-proportional and proportional feedbacks. The behaviors of cooperative opinion dynamics are analyzed using a null space property of state-dependent matrix-weighted Laplacian matrices and a Lyapunov candidate. Various consensus properties of state-dependent matrix-weighted Laplacian matrices are predicted according to the intra-agent network topology and interdependency topical coupling topologies. 
\end{abstract}

\section{Introduction}
The problem of opinion dynamics has attracted a lot of attention recently due to its applications to decision-making processes and evolution of public opinions \cite{Jiang_etal_tsmcs_2018}. The opinion dynamics arises between persons who interact with each other to influence others' opinions or to update his or her opinion \cite{Cho_tcss_2018}. The opinion dynamics has been also studied in control territory or in signal processing recently. For examples,  control via leadership with state and time-dependent interactions \cite{Dietrich_etal_tac_2018}, game theoretical analysis of  the Hegselmann-Krause Model \cite{Etesami_Sasar_tac_2015}, Hegselmann-Krause dynamics for the continuous-agent model  \cite{Wedin_Hegarty_tac_2015}, and the impact of random actions  \cite{Leshem_Scaglione_tsipn_2018} 
have been investigated. The opinion dynamics under consensus setups has been also studied \cite{Semonsen_etal_tcyb_2018,Dong_etal_ieeebigdata_2018}. In opinion dynamics under scalar-based consensus laws, the antagonistic interactions in some edges are key considerations \cite{Altafini_tac_2013,Hendrickx_cdc_2014,Proskurnikov_etal_tac_2016}. The antagonistic interactions may represent repulsive or anti-cooperative characteristics between neighboring agents. In traditional consensus, all the interactions between agents are attractive one; so the dynamics of the traditional consensus has a contraction property, which eventually ensures a synchronization of agents. However, if there is an antagonistic interaction, a consensus may not be achieved and the Laplacian matrix may have negative eigenvalues \cite{Boyd_icm_2006}. Thus, in the existing opinion dynamics, the antagonistic interactions are modeled such that the Laplacian matrix would not have any negative eigenvalues. Specifically, in \cite{Altafini_tac_2013}, signs of adjacent weights are used to model antagonistic interactions resulting in Laplacian matrix with absolute diagonal elements, and in \cite{Hendrickx_cdc_2014}, the author has extended the model of \cite{Altafini_tac_2013} to the one that allows arbitrary time-dependent interactions. In \cite{Proskurnikov_etal_tac_2016}, they have further considered time-varying signed graphs under the setup of the antagonistic interactions. On the other hand, opinion dynamics with  state constraints was also examined when the agents are preferred to attach to the initial opinion, i.e., with stubborn agents \cite{Cao_acc_2015}. Recently, in \cite{Liu_etal_tac_2018}, they have examined a joint impact of the dynamical properties of individual agents and the interaction topology among them on polarizability, consensusability, and neutralizability, with a further extension to heterogeneous systems with non-identical dynamics.

Unlike the scalar-consensus based updates, there also have been some works on opinion dynamics with matrix weighed interactions. Recently, opinion dynamics with multidimensional or multiple interdependent topics have been reported in \cite{Parsegov_etal_tac_2017,Mengbin_etal_auto_sub_2017}. In \cite{Parsegov_etal_tac_2017}, multidimensional opinion dynamics based on Friedkin and Johnsen (FJ) model and DeGroot models were analyzed in the discrete-time domain. The continuous-time version of \cite{Parsegov_etal_tac_2017} with stubborn agents was presented and analyzed in \cite{Mengbin_etal_auto_sub_2017}. The DeGroot-Friedkin model was also analyzed to conclude that it has an exponential convergent equilibrium point  \cite{Mengbin_etal_tac_2018}. Also in  \cite{Mengbin_etal_tac_2018}, they considered the dynamic network topology to evaluate the propagation property of the social power. Since the topics are interdependent and coupled with each other, these works may be classified as matrix-weighted consensus problems \cite{Minh_etal_auto_2018}. Opinion dynamics under leader agents with matrix weighted couplings was studied in \cite{Minh_etal_ascc_2017}.

In this paper, we would like to present a new model for opinion dynamics on multiple interdependent topics under a state-dependent matrix weighted consensus setup. We first provide a model for characterizing the coupling effects of multiple interdependent topics. We consider both the proportional and inverse proportional feedback effects on diagonal and off-diagonal terms. The cooperative dynamics and non-cooperative dynamics are modeled using the signs of diffusive couplings of each topic. Then, we provide some analysis on the convergence or consensus of the topics. Two results will be presented according to the property of weighting matrices. The first result is developed when the coupling matrices are positive semidefinite. When the coupling matrices are positive semidefinite, exact conditions for complete opinion consensus and cluster consensus are provided. Then, as the second result, when the coupling matrices are indefinite, we provide a sufficient condition for a complete opinion consensus. 

Consequently, the main contributions of this paper can be summarized as follows. First, a model for opinion dynamics is established. The connectivities are characterized by interaction topology between agents and coupling topology among the topics. Thus, the overall system has two-layer network topologies. Second, analysis for complete opinion consensus and partial opinion consensus is presented for both the cases when the coupling matrices are positive semidefinite and indefinite. As far as the authors are concerned, this is the first paper that presents a detailed model for inverse-proportional and proportional feedback opinion dynamics along with the convergence analysis. 
This paper is organized as follows. Section~\ref{section_model} provides a detailed process for building models for opinion dynamics. Section~\ref{section_analysis} presents the analysis for convergence of cooperative opinion dynamics. Section~\ref{section_sim} is dedicated to simulation results and Section~\ref{section_con} concludes this paper with some discussions.


\section{Modeling} \label{section_model}
There are $d$ different topics that may be of interests to the members of a society. Let the set of topics be denoted as $\mathcal{T}=\{1, \ldots, d\}$ and let the opinion vector associated with the member $i$  be written as $x_{i} = (x_{i,1}, x_{i,2}, \ldots, x_{i,d})^T$. We can write the $i$-th agent's opinion about the $p$-th topic as $x_{i,p}$. Each member (or can be called agent) has its initial opinion on the topics as $x_{i,k}(t_0) = (x_{i,1}(t_0), x_{i,2}(t_0), \ldots, x_{i,d}(t_0))^T$. The opinion dynamics of agent $i$ can be modeled as
\begin{align}
\left(\begin{array}{c} \dot{x}_{i,1} \\ \dot{x}_{i,2} \\ \vdots \\ \dot{x}_{i,d} \end{array} \right) &=\sum_{j \in \mathcal{N}_i}^n \left[\begin{array}{ccc}   
a_{1,1}^{i,j}   & \ldots  & a_{1,d}^{i,j}  \\
a_{2,1}^{i,j}   & \ldots  & a_{2,d}^{i,j}  \\
\vdots        &   \ddots  &   \vdots  \\
a_{d,1}^{i,j}  & \ldots  & a_{d,d}^{i,j}  \\
\end{array} \right]  \left(\begin{array}{c} {x}_{j,1} - {x}_{i,1} \\ {x}_{j,2} - {x}_{i,2} \\ \vdots \\ {x}_{j,d} - {x}_{i,d} \end{array} \right) \nonumber \\
\triangleq \dot{x}_i &= \sum_{j \in \mathcal{N}_i}^n A^{i,j} (x_j - x_i) \label{eq_agent_dynamics}
\end{align} 
where $A^{i,j} \in \Bbb{R}^{d \times d}$ is the matrix weighting for the edge $(i,j) \in \mathcal{E}$ and $i \in \mathcal{V}$, and $\vert \mathcal{E} \vert =m$ and $\vert \mathcal{V} \vert =n$. The practical meaning of \eqref{eq_agent_dynamics} is that each member of a society may have its own opinion about the topics, and the opinions are inter-coupled with the opinions of the neighboring agents. Thus, the matrix $A^{i,j}$ characterizes the logical reasoning of agent $i$ with opinions from agent $j$. The neighborhoods of agents are determined by the interaction graph $\mathcal{G}=(\mathcal{V}, \mathcal{E})$. If a topic in member $i$ has at least one connection to another topic or the same topic of another agent $j$, then two agents $i$ and $j$ are called \textit{connected}. The terminology \textit{connection} or \textit{connected} is used for defining the connection in the level of agents. When there are connections between the topics, it is called \textit{coupled} or \textit{couplings} between topics.  Thus, the terminology \textit{coupled} or \textit{couplings} is used in the level of topics. Therefore, based on the terminological definitions, if there is at least one coupling between the topics of agents $i$ and $j$, then two agents $i$ and $j$ can be considered as connected. However, even though two agents are connected, it does not mean that a topic in an agent is connected to another topic of the other agent. The formal definitions are given as follows.
\begin{definition}
Two agents $i$ and $j$ are considered connected if $A^{i,j}$ is not identically zero, i.e., $A^{i,j} \neq 0$. The topology for overall network connectivities is represented by the interaction graph $\mathcal{G}=(\mathcal{V}, \mathcal{E})$ where the edge set $\mathcal{E}$ characterizes the connectivities between agents. 
If there is a spanning tree in the network $\mathcal{G}$, it is called connected.  
For a topic $p \in \mathcal{T}$, the graph is called $p$-coupled if the elements of the set $\{a_{p,p}^{i,j}, ~\forall (i,j) \in \mathcal{E}\}$ are connected for the topic $p$. The topology for the topic $p$ is defined by the graph $\mathcal{G}_p=(\mathcal{V}_p, \mathcal{E}_p)$, where $p \in \mathcal{T}$, and $\mathcal{V}_p= \{x_{1,p}, x_{2,p}, \ldots, x_{n,p}\}$ and $\mathcal{E}_p = \{ (i,j)~:~ 
a_{p,p}^{i,j} \neq 0 \}$. If it is $p$-coupled for all topics $p \in \mathcal{T}$, it is called all-topic coupled.
\end{definition}
\begin{definition}
For the edge $(i,j)$, let the topology for the couplings among topics be denoted as $\mathcal{G}_{i,j}= (\mathcal{V}_{i,j}, \mathcal{E}_{i,j})$, which is called coupling graph for the edge $(i,j)$, where $\mathcal{V}_{i,j}$ includes all the topics contained in the agents $i$ and $j$, and $\mathcal{E}_{i,j}$ includes all the couplings. If $\mathcal{G}_{i_1,j_1}=\mathcal{G}_{i_2,j_2}$ for all edges $(i_1, j_1) \neq (i_2, j_2)$, then all the coupling topologies of the society are equivalent. If all the coupling topologies between agents are equivalent, it is called homogeneous-coupling network. Otherwise, it is called heterogeneous-coupling network. 
\end{definition}
Based on the above definitions, we can see that every $\mathcal{G}_p$ is disconnected even though $\mathcal{G}$ is connected. If the union of all $\mathcal{G}_p$ is connected, then $\mathcal{G}$ is also connected. Also, since each agent has the same set of topics,  $\mathcal{V}_{i,j} = \mathcal{T}$ for all $(i,j) \in \mathcal{E}$.

\begin{assumption}The coupling between neighboring agents is symmetric, i.e., if there exists a coupling $(p, q)$ in $\mathcal{E}_{i,j}$, there also exists a coupling $(q,p)$ in $\mathcal{E}_{i,j}$.
\end{assumption}

\begin{example}
Fig.~\ref{connected_coupled} shows the concepts of ``\textit{connected}'' and ``\textit{coupled}'' in neighboring agents. The coupling graph $\mathcal{G}_{i,j}$ can be determined as $\mathcal{G}_{i,j} = (\mathcal{V}_{i,j}, \mathcal{E}_{i,j})$ where $\mathcal{V}_{i,j} =\{p, q, r\}$ and $\mathcal{E}_{i,j} = \{ (p,q), (q, r), (q,p), (r, q) \}$. 
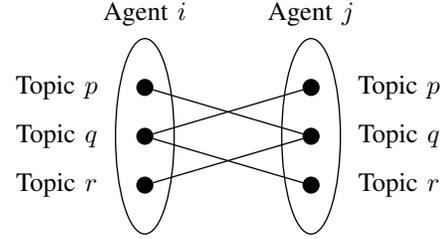
\begin{figure}
\centering
\begin{tikzpicture}[scale=0.65]

\node[place, black] (node21) at (-2.2,3) [label=left:$$] {};
\node[place, black] (node22) at (-2.2,2) [label=left:$$] {};
\node[place, black] (node23) at (-2.2,1) [label=left:$$] {};

\node[place, black] (node31) at (1.2,3) [label=left:$$] {};
\node[place, black] (node32) at (1.2,2) [label=left:$$] {};
\node[place, black] (node33) at (1.2,1) [label=left:$$] {};


\draw[black, line width=.5pt] (-2.2,2) ellipse (0.6cm and 2cm);
\draw[black, line width=.5pt] (1.2,2) ellipse (0.6cm and 2cm);

\node[black] at (-2.2, 4.5) {Agent $i$};
\node[black] at (1.2, 4.5) {Agent $j$};

\node[black] at (-4.0, 3) {Topic $p$};
\node[black] at (-4.0, 2) {Topic $q$};
\node[black] at (-4.0, 1) {Topic $r$};

\node[black] at (3.0, 3) {Topic $p$};
\node[black] at (3.0, 2) {Topic $q$};
\node[black] at (3.0, 1) {Topic $r$};



\draw (node21) [line width=0.5pt] -- node [left] {} (node32);
\draw (node22) [line width=0.5pt] -- node [right] {} (node33);

\draw (node22) [line width=0.5pt] -- node [left] {} (node31);
\draw (node23) [line width=0.5pt] -- node [right] {} (node32);



\end{tikzpicture}
\caption{Connected vs. Coupled: Topics $p$ and $q$, and $q$ and $r$ are coupled in the coupling graph $\mathcal{G}_{i,j}$; so the agents $i$ and $j$ are connected. But, although the agents $i$ and $j$ are connected, for example, the topics $p$ and $r$ are not coupled.}
\label{connected_coupled}
\end{figure}
\end{example}

Each agent updates its coefficients in the matrix $A^{i,j}$ in the direction of cooperation or in the direction of antagonism. For a cooperative update, the rules for opinion update are formulated as:
\begin{itemize}
\item The diagonal terms: If $a_{k,k}^{i,j}$ is positive and as it increases, the tendency of agreement between $x_{j,k}$ and $x_{i,k}$ increases. Otherwise, if $a_{k,k}^{i,j}$ is negative and as it increases to bigger negative value, the tendency of anti-agreement between $x_{j,k}$ and $x_{i,k}$ becomes significant. 
\item The off-diagonal terms: Let us consider the effect of $a_{2,1}^{i,j}$. We can consider the following four cases:
\begin{enumerate}
\item Case $1$: $(x_{j,2}- x_{i,2}) \geq 0$  and $(x_{j,1}- x_{i,1}) \geq 0$
\item Case $2$: $(x_{j,2}- x_{i,2}) \geq 0$  and $(x_{j,1}- x_{i,1}) < 0$
\item Case $3$: $(x_{j,2}- x_{i,2}) < 0$  and $(x_{j,1}- x_{i,1}) \geq 0$
\item Case $4$: $(x_{j,2}- x_{i,2}) < 0$  and $(x_{j,1}- x_{i,1}) < 0$
\end{enumerate}
When $(x_{j,2}- x_{i,2}) \geq 0$, agent $i$ needs to increase the value of $x_{i,2}$ to reach a consensus to $x_{j,2}$. Otherwise, if $(x_{j,2}- x_{i,2}) < 0$, agent $i$ needs to decrease the value of $x_{i,2}$ to reach a consensus to $x_{j,2}$. So, for the cases $1$ and $2$, to enhance the agreement tendency, it needs to increase the value of $x_{i,2}$, by way of multiplying $a_{2,1}^{i,j}$ and $(x_{j,1}- x_{i,1})$. Thus, when $(x_{j,1}- x_{i,1}) \geq 0$, we can select $a_{2,1}^{i,j} >0$; but when  $(x_{j,1}- x_{i,1}) < 0$, we can select $a_{2,1}^{i,j} <0$. On the other hand, in the case of  $(x_{j,2}- x_{i,2}) < 0$, we can select $a_{2,1}^{i,j} <0$ when $(x_{j,1}- x_{i,1}) \geq 0$, or we can select  $a_{2,1}^{i,j} > 0$ when $(x_{j,1}- x_{i,1}) < 0$. For the anti-consensus update, $a_{2,1}^{i,j}$ should be selected with opposite signs. 
\end{itemize}

The effects of diagonal terms can be modeled as follows: 
\begin{definition} \label{def_consensus_effects} Direct coupling effects in diagonal terms:
\begin{itemize}
\item Proportional feedbacks: A close opinion between two agents acts as for increasing the consensus tendency between them.
\item Inverse proportional feedbacks: A quite different opinion between two agents acts as for increasing the consensus tendency between them.
\end{itemize}
\end{definition}
The off-diagonal terms need to be designed carefully taking account of the coupling effects in different topics. 
\begin{definition} \label{def_coupling_effects} Cross coupling effects in off-diagonal terms:
\begin{itemize}
\item Proportional feedbacks: A close opinion in one topic acts as for increasing the consensus tendency of other topics.
\item Inverse proportional feedbacks: A quite different opinion in one topic acts as for increasing the consensus tendency of other topics.
\end{itemize}
\end{definition}

\begin{definition} (Completely and partial opinion consensus, and clusters) If a consensus is achieved for all topics, i.e., $x_{j,p} = x_{i,p}$ for all $p \in \mathcal{T}$, it is called a complete opinion consensus. In this case, there exists only one cluster. Otherwise, if a part of topics is agreed, it is called partial opinion consensus. When only a partial opinion consensus is achieved, there could exist clusters $\mathbb{C}_k, k=1, \ldots, q$ such that $\mathbb{C}_i \cap \mathbb{C}_j = \emptyset$ for different $i$ and $j$, and $\sum_{k=1}^q \mathbb{C}_k = \{x_1, x_2, \ldots, x_n\}$, and in each cluster, $x_i = x_j$ when $x_i$ and $x_j$ are elements of the same cluster, i.e., $x_i, x_j \in  \mathbb{C}_k$. \label{def_com_par_opinion_con}
\end{definition}

\begin{definition} (Complete clustered consensus)
If the opinions of agents are completely divided without ensuring any partial opinion consensus between them, it is called completely clustered consensus. 
\end{definition}

In the case of a partial opinion consensus as per \textit{Definition~\ref{def_com_par_opinion_con}}, the clusters are not completely divided clusters, i.e., in two different clusters, some topics may reach a consensus. 
\begin{example}
In Fig.~\ref{partial_complete_consensus}, there are five agents, with three topics. The agents reach a consensus on the topic $p=3$. But, on other topics $p=1, 2$, they do not reach a consensus. For the topic $1$, there are two clusters (i.e., agents $1,2,3$ in one cluster, and agents $4, 5$ in another cluster), and for the topic $2$, there are also two clusters (i.e., agents $1$ in one cluster, and agents $2, 3, 4, 5$ in another cluster). So, overall, the network has a consensus in a part of topics, but they do not reach a consensus on the other topics. So, no complete opinion consensus is achieved, and no complete clustered consensus is achieved. Consequently, there are three clusters as $\mathbb{C}_1 =\{x_1\}$, $\mathbb{C}_2 =\{x_2, x_3\}$, and $\mathbb{C}_3 =\{x_4, x_5\}$.
\begin{figure}
\centering
\begin{tikzpicture}[scale=0.65]
\node[place2, black] (node11) at (-3.6,3) [label=left:$ $] {};
\node[place, black] (node12) at (-3.6,2) [label=left:$ $] {};
\node[place2, black] (node13) at (-3.6,1) [label=left:$ $] {};

\node[place2, black] (node21) at (-1.2,3) [label=left:$ $] {};
\node[place2, black] (node22) at (-1.2,2) [label=left:$ $] {};
\node[place2, black] (node23) at (-1.2,1) [label=left:$ $] {};

\node[place2, black] (node31) at (1.2,3) [label=left:$ $] {};
\node[place2, black] (node32) at (1.2,2) [label=left:$ $] {};
\node[place2, black] (node33) at (1.2,1) [label=left:$ $] {};

\node[place, black] (node41) at (3.6,3) [label=left:$ $] {};
\node[place2, black] (node42) at (3.6,2) [label=left:$ $] {};
\node[place2, black] (node43) at (3.6,1) [label=left:$ $] {};

\node[place, black] (node51) at (6.0,3) [label=left:$ $] {};
\node[place2, black] (node52) at (6.0,2) [label=left:$ $] {};
\node[place2, black] (node53) at (6.0,1) [label=left:$ $] {};

\draw[black, line width=.5pt] (-3.6,2) ellipse (0.6cm and 2cm) ;
\draw[black, line width=.5pt] (-1.2,2) ellipse (0.6cm and 2cm);
\draw[black, line width=.5pt] (1.2,2) ellipse (0.6cm and 2cm);
\draw[black, line width=.5pt] (3.6,2) ellipse (0.6cm and 2cm);
\draw[black, line width=.5pt] (6.0,2) ellipse (0.6cm and 2cm);

\draw[dashed] (-4.6,5.0) rectangle (-2.5,-0.1);
\draw[dashed] (-2.3,5.0) rectangle (2.3,-0.1);
\draw[dashed] (2.5,5.0) rectangle (7.1,-0.1);


\node[black] at (-3.6, 4.5) {Agent 1};
\node[black] at (-1.2, 4.5) {Agent 2};
\node[black] at (1.2, 4.5) {Agent 3};
\node[black] at (3.6, 4.5) {Agent 4};
\node[black] at (6.0, 4.5) {Agent 5};

\node[black] at (-3.6, 5.5) {$\mathbb{C}_1$};
\node[black] at (0, 5.5) {$\mathbb{C}_2$};
\node[black] at (4.8, 5.5) {$\mathbb{C}_3$};

\node[black] at (-5.6, 3) {Topic 1};
\node[black] at (-5.6, 2) {Topic 2};
\node[black] at (-5.6, 1) {Topic 3};


\draw (node11) [line width=1.0pt] -- node [left] {} (node21);
\draw (node13) [line width=1.0pt] -- node [below] {} (node23);

\draw (node21) [line width=1.0pt] -- node [left] {} (node31);
\draw (node22) [line width=1.0pt] -- node [right] {} (node32);
\draw (node23) [line width=1.0pt] -- node [left] {} (node33);

\draw (node32) [line width=1.0pt] -- node [right] {} (node42);
\draw (node33) [line width=1.0pt] -- node [below] {} (node43);

\draw (node41) [line width=1.0pt] -- node [left] {} (node51);
\draw (node42) [line width=1.0pt] -- node [right] {} (node52);
\draw (node43) [line width=1.0pt] -- node [below] {} (node53);

\end{tikzpicture}
\caption{Partial opinion consensus and clusters. The topic $3$ reaches a consensus, while topics $1$ and $2$ do not reach a consensus.}
\label{partial_complete_consensus}
\end{figure}
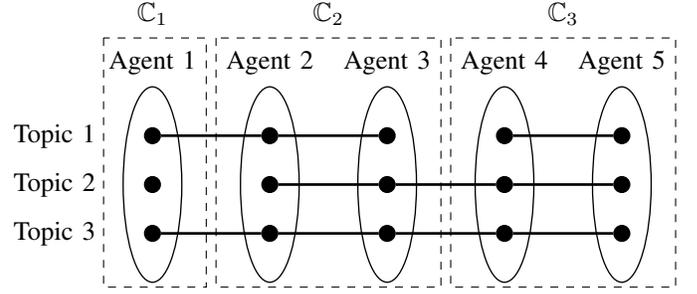
\end{example}

The consensus and coupling effects given in the \textit{Definition~\ref{def_consensus_effects}} and \textit{Definition~\ref{def_coupling_effects}} can be mathematically modeled as follows. 

\subsubsection{Inverse-proportional feedbacks} When the values of opinions of two agents are quite different, the coupling effects are more significant, which may be against from a natural phenomenon (ex, gravitational force). That is, when two opinions are close, there could be more attraction force. In inverse-proportional feedbacks, there will be more coupling effects when the values of opinions are quite different. 
\begin{itemize}
\item Direct coupling in diagonal terms:
\begin{align}  \label{inverse_proportional_diagonal_dynamics}
a_{p,p}^{i,j} &= k_{p,p}^{i,j} 
\end{align}
where $k_{p,p}^{i,j} =  k_{p,p}^{j,i} >0$. 

\item Cross coupling in off-diagonal terms:
\begin{align}  \label{inverse_proportional_coupling_dynamics}
a_{p,q}^{i,j} &= k_{p,q}^{i,j} \times   \text{sign}( x_{j,p} - x_{i,p}) \times  \text{sign}(x_{j,q} - x_{i,q}) 
\end{align}
where $k_{p,q}^{i,j} =  k_{p,q}^{j,i} = k_{q,p}^{i,j} >0$, and $\text{sign}( x_{j,p} - x_{i,p})=1$ when $x_{j,p} - x_{i,p} \geq 0$ and $\text{sign}( x_{j,p} - x_{i,p})=-1$ when $x_{j,p} - x_{i,p} < 0$. 
\end{itemize}

\subsubsection{Proportional feedbacks}  In proportional feedbacks, there will be less coupling effects when the values of opinions are quite different. 
\begin{itemize}
\item Direct coupling in diagonal terms:
\begin{align} \label{proportional_diagonal_dynamics}
a_{p,p}^{i,j} = \frac{  k_{p,p}^{i,j} }{ c_2 \Vert x_{j,p} - x_{i,p} \Vert^2 + c_1 \Vert x_{j,p} - x_{i,p} \Vert  + c_0  }
\end{align}
where $k_{p,p}^{i,j} =  k_{p,p}^{j,i} >0$. 

\item Cross coupling in off-diagonal terms:
\begin{align} \label{proportional_coupling_dynamics2}
a_{p,q}^{i,j} = \frac{  k_{p,q}^{i,j} \times   \text{sign}( x_{j,p} - x_{i,p}) \times  \text{sign}(x_{j,q} - x_{i,q})   }{ (c_{1} \Vert x_{j,p} - x_{i,p} \Vert + c_{0}) (c_{1} \Vert x_{j,q} - x_{i,q} \Vert + c_{0})  } 
\end{align}
where $k_{p,q}^{i,j} =  k_{q,p}^{i,j} >0$ and $k_{p,q}^{i,j} =  k_{p,q}^{j,i} >0$, and $c_{1}$ and $c_{0}$ are positive constants. Then, we can have $(A^{i,j})^T = A^{i,j}$ and $A^{i,j}= A^{j,i}$. Thus, with the model \eqref{proportional_coupling_dynamics2}, the Laplacian matrix is of symmetric. 
\end{itemize}

Note that in the above coupling models, if $(p,q) \in \mathcal{E}_{i,j}$, then $a_{p,q}^{i,j}\neq 0$, otherwise, $a_{p,q}^{i,j}=0$. So, the matrix $A^{i,j}=[a_{p,q}^{i,j}]$ is the weighting matrix for the topics between two agents $i$ and $j$. But, the matrix $A^{i,j}$ is state- and sign-dependent, while the matrix $K^{i,j}=[k_{p,q}^{i,j}]$ is a matrix which defines the topological characteristics between topics of the neighboring agents. The matrix $K^{i,j}$ is called  coupling matrix, and it is a constant matrix. 

It is remarkable that the matrix $A^{i,j}$ is not an adjacency matrix, and neither is the matrix $K^{i,j}$. But, they are similar to an adjacency matrix. For example, if there is no direct coupling  between the same topics, then $K^{i,j}$ is the adjacency matrix for characterizing the couplings between the topics of two neighboring agents. 
On the other hand, if all the topics are coupled (i.e., a topic of agent $i$ is coupled to all the topics of neighboring agent $j$), then $K^{i,j} - \mathbb{I}_d$ is the adjacency matrix ignoring the self-loops. The direct coupling in the $p$-th topic implies that there is a self-loop in the $p$-th topic node. 
\begin{example} Let us consider the following coupling matrices.
\begin{align}
K_1^{i,j} =\left[\begin{array}{ccccc}   
2 &  1  &  1  & 1  & 1  \\
1 &  2  &  1  & 1  & 1  \\
1 &  1  &  2  & 1  & 1  \\
1 &  1  &  1  & 3  & 1  \\
1 &  1  &  1  & 1  & 1  \\
\end{array} \right], ~ K_2^{i,j} =\left[\begin{array}{ccccc}   
0 &  1  &  1  & 1  & 1  \\
1 &  0  &  1  & 1  & 1  \\
1 &  1  &  0  & 1  & 1  \\
1 &  1  &  1  & 0  & 1  \\
1 &  1  &  1  & 1  & 0  \\
\end{array} \right] \nonumber
\end{align} 
The matrix $K_1^{i,j}$ means that any topic in agent $i$ is coupled to all the topics in $j$, while the matrix $K_2^{i,j}$ means that any topic in agent $i$ is coupled to all the topics in $j$, but $x_{i,p}$ is not coupled to $x_{j,p}$ (i.e., no direct coupling). 
\end{example}

The anti-consensus can be simply modeled by adding the minus sign to the elements of the coupling matrix, i.e., $k_{p,q}^{i,j}$. Thus, there are four types of couplings: proportional coupling, proportional anti-coupling, inverse-proportional coupling, and inverse-proportional anti-coupling. The dynamics with anti-consensus terms is called \textit{non-cooperative opinion dynamics}, while the dynamics without anti-consensus terms is called \textit{cooperative opinion dynamics}.
Note that in existing traditional consensus works, the inverse proportional diagonal terms, i.e., \eqref{inverse_proportional_diagonal_dynamics}, are only used for the consensus couplings. 


The dynamics \eqref{eq_agent_dynamics} can be concisely rewritten as:
\begin{align} \label{eq_laplacian_dynamics}
\dot{x} = - \mathbb{L}(x_1, \ldots, x_n) x
\end{align}
where the Laplacian is computed as
\begin{align}
\mathbb{L} = \left[\begin{array}{cccc}
\sum_{j \in \mathcal{N}_i} A^{1,j}   & - A^{1,2} & \ldots      & - A^{1,n} \\
- A^{2,1}   & \sum_{j \in \mathcal{N}_i} A^{2,j}   & \ldots      & - A^{2,n} \\
\vdots   & \vdots & \ddots &  \vdots \\
- A^{n,1}   & - A^{n,2}   & \ldots      & \sum_{j \in \mathcal{N}_i} A^{n,n} \\
\end{array}\right] \label{symbol_Laplacian}
\end{align}
Note that in the dynamics \eqref{eq_laplacian_dynamics}, the Laplacian  $\mathbb{L}$ is dependent upon the sign of $x_{j,p}-x_{i,p}$; thus, it is discontinuous when the sign changes abruptly. To be a continuous function, the sign function may be modified as a sigmoid function as:
\begin{align}
sign(x_{j,p} - x_{i,p}) \triangleq \frac{ 2 }{ 1 + e^{-k_e (x_{j,p} - x_{i,p})}} - 1
\end{align}
where $k_e$ is a sufficiently large positive constant. We remark that the sign function can be also changed to a signum function
$\text{sig}( x_{j,p} - x_{i,p}) |x_{j,p} - x_{i,p}|^\alpha$, where $0 < \alpha < 1$.

In the case of the inverse-proportional feedback laws, we can see that $a_{p,q}^{i,j} =  a_{q,p}^{i,j}$, and $a_{p,q}^{i,j} =  a_{p,q}^{j,i}$. Then,  the Laplacian matrix $\mathbb{L}$ is of symmetric.
Consequently, for the inverse-proportional consensus couplings, we can rewrite \eqref{eq_agent_dynamics} as:
\begin{align}
\left(\begin{array}{c} \dot{x}_{i,1} \\ \dot{x}_{i,2} \\ \vdots \\ \dot{x}_{i,d} \end{array} \right) &=\sum_{j \in \mathcal{N}_i}^n \left[\begin{array}{ccc}   
\text{sgn}_{1,1}^{i,j} k_{1,1}^{i,j} & \ldots  & \text{sgn}_{1,d}^{i,j} k_{1,d}^{i,j}  \\
\text{sgn}_{2,1}^{i,j} k_{2,1}^{i,j} &  \ldots  & \text{sgn}_{2,d}^{i,j} k_{2,d}^{i,j}  \\
\vdots      &   \ddots  &   \vdots  \\
\text{sgn}_{d,1}^{i,j} k_{d,1}^{i,j} & \ldots  & \text{sgn}_{d,d}^{i,j} k_{d,d}^{i,j}  \\
\end{array} \right] \nonumber\\
&~~~~~ \times \left(\begin{array}{c}  {x}_{j,1} - {x}_{i,1}  \\  {x}_{j,2} - {x}_{i,2}  \\ \vdots \\  {x}_{j,d} - {x}_{i,d} \end{array} \right)  \label{eq_inverse_pro_consensus} 
\end{align} 
where $\text{sgn}_{p,q}^{i,j} \triangleq  \text{sign}( x_{j,p} - x_{i,p}) \times  \text{sign}(x_{j,q} - x_{i,q})$ and $\text{sgn}_{p,q}^{i,j} = \text{sgn}_{q,p}^{i,j}$. If there are some inverse-proportional anti-consensus couplings between some topics, then some elements in \eqref{eq_inverse_pro_consensus} will have negative signs. For example, if the $1$-st topic and $2$-nd topic are anti-consensus coupled, then the terms $\text{sgn}_{1,2}^{i,j} k_{1,2}^{i,j}$ and $\text{sgn}_{2,1}^{i,j} k_{2,1}^{i,j}$ need to be modified as $-\text{sgn}_{1,2}^{i,j} k_{1,2}^{i,j}$ and $-\text{sgn}_{2,1}^{i,j} k_{2,1}^{i,j}$. 
But, in this case, the Laplacian matrix $\mathbb{L}$ may have negative eigenvalues; thus, the stability or convergence may not be ensured any more. Thus, in this paper, we focus on only the cooperative opinion dynamics. For a matrix $A$, we use $\mathcal{N}(A)$ and $\mathcal{R}(A)$ to denote the nullspace and the range of $A$, respectively.


\section{Analysis} \label{section_analysis}
It will be shown in this section that the positive definiteness of the Laplacian matrix in \eqref{symbol_Laplacian} is closely related with the positive definiteness of the coupling matrix $K^{i,j}$. When the coupling matrices are positive semidefinite, we provide exact conditions for complete opinion consensus and cluster consensus. However, when $\mathbb{L}(x)$ is indefinite, since $\mathbb{L}(x)$ is time-varying, the system \eqref{eq_laplacian_dynamics} can still be stable and a consensus might be reached. Let us first focus on the case of positive semidefinite Laplacian, and then, we consider general cases that include indefinite Laplacian matrices.

\subsection{Case of Positive Semidefinite Laplacian} \label{sec3_sub1_psd}
It is not straightforward to verify whether the Laplacian $\mathbb{L}(x)$ in \eqref{symbol_Laplacian} is positive semidefinite or not since it is a block matrix. In $\mathbb{L}(x)$, the element matrices could be positive definite, positive semidefinite, negative definite, negative semidefinite, or indefinite. Thus, an analysis for the dynamics \eqref{eq_laplacian_dynamics} would be more difficult than the traditional scalar-based consensus. For the analysis, let us define the incidence matrix $\mathbb{H}= [h_{ij}] \in \Bbb{R}^{m \times n}$ for the interaction graph $\mathcal{G}=(\mathcal{V}, \mathcal{E})$ as:
\begin{align}
h_{ki} = \begin{cases}
       -1  & ~ \text{if vertex $i$ is the tail of the $k$-th edge}     \\
       1  & ~ \text{if vertex $i$ is the head of the $k$-th edge}   \\
       0  & ~\text{otherwise}
       \end{cases}
\end{align}
where the direction of the edge $k$ is arbitrary. Let us also define the incidence matrix in $d$-dimensional space as $\bar{\mathbb{H}} = \mathbb{H} \otimes \mathbb{I}_d$ and write the weighting matrix for the $k$-th edge as $A^{k\text{-th}} \in \Bbb{R}^{d \times d}$. Let us also write the coupling matrix $K^{i,j}$ corresponding to $A^{k\text{-th}}$ as $K^{k\text{-th}}$. 
As aforementioned, if there is no direct coupling between the same topics of two neighboring agents, the coupling matrix $K^{i,j}$ can be considered as a constant adjacency matrix for the coupling graph $\mathcal{G}_{i,j}$. The block diagonal matrix composed of $A^{k\text{-th}}, ~k=1, \ldots, m$ is denoted as $\text{blkdg}(A^{k\text{-th}})$ and the block  diagonal matrix composed of $K^{k\text{-th}}, ~k=1, \ldots, m$ is denoted as $\text{blkdg}(K^{k\text{-th}})$.
\begin{lemma} \label{lemma_psd}
For the inverse proportional coupling, the Laplacian $\mathbb{L}(x)$ is positive semidefinite if and only if  $\text{blkdg}(K^{k\text{-th}})$ is positive semidefinite. 
\end{lemma}
\begin{proof}
Due to the same reason as Lemma 1 of \cite{Trinh2017arvix}, we can write $\mathbb{L} = \bar{\mathbb{H}}^T \text{blkdg}(A^{k\text{-th}}) \bar{\mathbb{H}}$. It is shown that the weighting matrix $A^{k\text{-th}}$ can be written as
\begin{align}
A^{k\text{-th}} = \text{diag}(S^{i,j}) K^{k\text{-th}} \text{diag}(S^{i,j}) \label{weighting_matrix_decompose}
\end{align}
where the edge $(i,j)$ is the $k$-th edge and $\text{diag}(S^{i,j})$ is given as
\begin{align}
\text{diag}(S^{i,j}) &= \text{diag}(\text{sign}(x_{j,k} - x_{i,k})) \nonumber\\
&= \left[\begin{array}{ccc}
\text{sign}(x_{j,1} - x_{i,1})   & \cdots   & 0 \\
\vdots   & \ddots  & \vdots  \\
0  &     \cdots & \text{sign}(x_{j,d} - x_{i,d})     \\
\end{array} \right] \label{weighting_matrix_S}
\end{align}
Hence, the Laplacian matrix can be written as
\begin{align}
\mathbb{L} = \bar{\mathbb{H}}^T \text{blkdg}( \text{diag}(S^{i,j})  )  \text{blkdg}(K^{k\text{-th}})  \text{blkdg}( \text{diag}(S^{i,j})  ) \bar{\mathbb{H}} 
\end{align}
Therefore, it is obvious that the Laplacian matrix $\mathbb{L}$ is positive semidefinite if and only if the matrix $\text{blkdg}(K^{k\text{-th}})$ can be decomposed as $\text{blkdg}(K^{k\text{-th}}) = \bar{\mathbb{K}}^T \bar{\mathbb{K}}$ with a certain matrix $\bar{\mathbb{K}}$. It means that the Laplacian matrix $\mathbb{L}$ is positive semidefinite if and only if $\text{blkdg}(K^{k\text{-th}})$ is positive semidefinite. 
\end{proof}

\begin{theorem}\label{thm_psd}
The Laplacian $\mathbb{L}(x)$ is positive semidefinite if and only if the coupling matrices $K^{i,j}$ are positive semidefinite. 
\end{theorem}
\begin{proof}
The proof is immediate from \textit{Lemma~\ref{lemma_psd}}.
\end{proof}

If is well-known that the adjacency matrix of a complete graph with $d$ nodes has eigenvalues $d-1$ with multiplicity $1$ and $-1$ with multiplicity $d-1$. Then, with this fact, we can obtain the following result.
\begin{theorem} Let us suppose that there is no direct coupling between the same topics of two neighboring agents; but a topic is coupled to all other topics. Then, under the condition that all $\text{diag}(S^{i,j})$ are not equal to zero (i.e., there exists at least one topic $p$ such that $x_{j,p} \neq x_{i,p}$), the Laplacian $\mathbb{L}(x)$ has negative eigenvalues. 
\end{theorem}
\begin{proof}
It is clear that the matrix $K^{k\text{-th}}$ can be considered as an adjacency matrix characterizing the topic couplings of the $k$-th edge. Let us denote this matrix as $K_{-}^{k\text{-th}}$. Then, the matrix $K_{-}^{k\text{-th}}$ has an eigenvalue $d-1$ with multiplicity $1$ and the eigenvalue $-1$ with multiplicity $d-1$. Thus, the matrix $\mathbb{L} = \bar{\mathbb{H}}^T \text{blkdg}( \text{diag}(S^{i,j})  )  \text{blkdg}(K_{-}^{k\text{-th}})  \text{blkdg}(\text{diag}(S^{i,j})) \bar{\mathbb{H}}$ has eigenvalues located in the open left half plane because $\text{blkdg}(K_{-}^{k\text{-th}})$ has eigenvalue $-1$ with multiplicity of $m (d-1)$ where $m =\vert \mathcal{E} \vert$. 
\end{proof}
The opposite circumstance occurs when all the topics are coupled, including the same topics, which is summarized  in the next result.
\begin{corollary}
Suppose that all the topics are coupled for all edges. Then, the Laplacian $\mathbb{L}(x)$ is positive semidefinite.
\end{corollary}
\begin{proof}
In this case, the matrix $K^{k\text{-th}}$ can be considered as a rank $1$ matrix defined as $K^{k\text{-th}} \triangleq  K_{+}^{k\text{-th}} =  K_{-}^{k\text{-th}} + \mathbb{I}_d$, because the matrix $K_{+}^{k\text{-th}}$ is a matrix with all elements being equal to $1$. Thus, the eigenvalues of $K_{+}^{k\text{-th}}$ are $d$ with multiplicity $1$ and all others being equal to zero.  Therefore, the matrix $\mathbb{L}$ positive semidefinite.
\end{proof}

Now, let us suppose that the matrix $\text{blkdg}(K^{k\text{-th}})$ is positive semidefinite; then it can be written as  $\text{blkdg}(K^{k\text{-th}}) = U^T U$ for some matrix $U$. Then, by denoting $\mathbb{U} = U \text{blkdg}( \text{diag}(S^{i,j})  )    \bar{\mathbb{H}}$, we can write $\mathbb{L}(x) = \mathbb{U}^T \mathbb{U}$. It is clear that $\text{nullspace}(\bar{\mathbb{H}}) \subseteq \text{nullspace}(\mathbb{L}) = \text{nullspace}(\mathbb{U})$, because $\text{blkdg}(K^{k\text{-th}})$ is positive semidefinite. Noticing that the null space of incidence matrix is $\mathcal{N}(\bar{\mathbb{H}}) = \mathcal{R}( \bf{1}_n \otimes \mathbb{I}_d) \triangleq \mathcal{R}$, we can see that the set $\mathcal{R}$ is always a subspace of $\mathcal{N}(\mathbb{L})$. To find the null space of $\mathbb{L}$, the following lemma will be employed. 
\begin{lemma} \cite{horn1990matrix} \label{lemma_lemma_psd_equal}
When a matrix $A$ is positive semidefinite, for any vector $x$, it holds that $A x = 0$ if and only if $x^T A x = 0$.
\end{lemma}
We remark that if a matrix $A$ is indefinite, the above lemma (i.e., \textit{Lemma~\ref{lemma_lemma_psd_equal}}) does not hold. For example, let us consider the following matrix:
\begin{align}
A = \left[\begin{array}{cccc}
1 & 0 & 0 & 0 \\
0 & 1 & 0 & 0 \\
0 & 0 & -4 & 0 \\
0 & 0 & 0 & -4 \\
\end{array}\right] \label{example_indefinite_null}
\end{align} 
which is nonsingular and has $1, 1, -4, -4$ as its eigenvalues. Then, the vector $x=(1, 1, 1/2, 1/2)^T$ makes that $x^T A x =0$, while $A x \neq 0$. It may be also important to note that, since the elements of the coupling matrix $K^{i,j}$ are all non-negative, and $\text{sign}(x_{j,p}- x_{i,p})(x_{j,p}- x_{i,p}) \geq 0$, \textit{Lemma~\ref{lemma_lemma_psd_equal}} may be further generalized. However, a further generalization is not obvious. 
From \textit{Lemma~\ref{lemma_lemma_psd_equal}}, we can see that a vector $x$ in the null space of $A$ is equivalent to a vector $x$ that makes $x^T A x =0$, when the matrix $A$ is positive semidefinite. With this fact, since  $a_{p,q}^{i,j} =  k_{p,q}^{i,j} \text{sgn}_{p,q}^{i,j} = k_{p,q}^{i,j}  \text{sign}( x_{j,p} - x_{i,p}) \times  \text{sign}(x_{j,q} - x_{i,q}) = k_{p,q}^{i,j} \left[ \frac{ 2 }{ 1 + e^{-k_e (x_{j,p} - x_{i,p})}} - 1 \right] \left[ \frac{ 2 }{ 1 + e^{-k_e (x_{j,q} - x_{i,q})}} - 1 \right]$, we can write $x^T \mathbb{L} x$ as follows:
\begin{align}
x^T \mathbb{L} x &= \sum_{(i,j) \in \mathcal{E} } (x_j  - x_i)^T A^{ij} (x_j  - x_i) \nonumber\\
        &= \sum_{(i,j) \in \mathcal{E} } \sum_{p=1}^d \sum_{q=1}^d a_{p,q}^{i,j} (x_{j,p} - x_{i,p})(x_{j,q} - x_{i,q})  \nonumber \displaybreak[3] \\
        &= \sum_{(i,j) \in \mathcal{E} } \sum_{p=1}^d \sum_{q=1}^d k_{p,q}^{i,j}  \nonumber \displaybreak[3] \\
        &~~~\times \left(\frac{ 2 }{ 1 + e^{-k_e (x_{j,p} - x_{i,p})}}  -1\right) ( x_{j,p} - x_{i,p}) \nonumber \displaybreak[3] \\
        &~~~\times \left( \frac{ 2 }{ 1 + e^{-k_e (x_{j,q} - x_{i,q})}}  -1\right) ( x_{j,q} - x_{i,q} ) \nonumber \displaybreak[3] \\
&= \sum_{(i,j) \in \mathcal{E}} \bm{\sigma}(x_j  - x_i)^T K^{i,j} \bm{\sigma}(x_j  - x_i) \label{eq:xT L x_null} 
\end{align}        
where $\bm{\sigma}(x_j  - x_i)$ is defined as
\begin{align}  
  \bm{\sigma}(x_j  - x_i) \triangleq &(\sigma(x_{j,1}  - x_{i,1}), \ldots, &\sigma(x_{j,d}  - x_{i,d}))^T
  \end{align}
with $\sigma(x_{j,p}  - x_{i,p})  \triangleq sign(x_{j,p}  - x_{i,p}) (x_{j,p} - x_{i,p})= \left\vert \frac{ 2 }{ 1 +e^{-k_e (x_{j,p} - x_{i,p})}}  -1 \right\vert \vert x_{j,p}  - x_{i,p} \vert$. 
Thus, to have $\sum_{(i,j) \in \mathcal{E} } \bm{\sigma}(x_j  - x_i)^T K^{i,j} \bm{\sigma}(x_j  - x_i)=0$, based on \textit{Lemma~\ref{lemma_lemma_psd_equal}}, it is required to satisfy  $\bm{\sigma}(x_j  - x_i) \in \mathcal{N}(K^{i,j})$ for all $(i,j) \in \mathcal{E}$. Therefore, summarizing this discussion, we can obtain the following lemma.
\begin{lemma} \label{lemma_nullspace_L}
Suppose that the $\text{blkdg}(K^{k\text{-th}})$ is positive semidefinite, which is equivalent to the positive semi-definiteness of $\mathbb{L}$. Then, the null space of $\mathbb{L}$ is given as:
\begin{align}
\mathcal{N}(\mathbb{L})=& \text{span}\{ \mathcal{R}, \{ x = (x_1^T, x_2^T, \cdots, x_n^T)^T \in \Bbb{R}^{dn} \nonumber\\
&~~~ ~|~   \bm{\sigma}(x_j  - x_i) \in \mathcal{N}(K^{i,j}),~\forall (i,j) \in \mathcal{E}   \}     \}
\end{align}
\end{lemma}

\begin{remark}
\textit{Lemma~\ref{lemma_nullspace_L}} implies that if $\text{blkdg}(K^{k\text{-th}})$ is positive definite, then 
$\mathcal{N}$$\mathcal(\mathbb{L}) = \mathcal{R}$. Thus, a complete opinion consensus is achieved.
\end{remark}

\begin{remark} \label{remark_nonsingular_result}
In \eqref{eq:xT L x_null}, if $K^{i,j}$ is nonsingular, then it has only the trivial null space. Thus, it appears that a complete opinion consensus might be achieved. However, as discussed with the $A$ matrix in \eqref{example_indefinite_null}, the set of vectors $x$ making $x^T \mathbb{L} x =0$ is not equivalent to the set of vectors $x$ making $\mathbb{L} x =0$. 
\end{remark}

In \textit{Lemma~\ref{lemma_nullspace_L}}, there are possibly two subspaces for the null space. The subspace  $\mathcal{R}$ is the standard consensus space; but the subspace spanned by $x$ satisfying $\bm{\sigma}(x_j  - x_i) \in \mathcal{N}(K^{i,j})$ needs to be elaborated since the elements of the coupling matrix $K^{i,j}$ are zero or positive constants, and the elements of the vector $\bm{\sigma}(x_j  - x_i)$ are also positive except the zero. The following example provides some intuitions for the coupling matrix. 
\begin{example}
Let us consider that there are five topics, and the coupling matrix between agents $i$ and $j$ is given as:
\begin{align}
K^{i,j} &=\left[\begin{array}{ccccc}   
0 &  1  &  0  & 0  & 0  \\
1 &  1  &  0  & 0  & 0  \\
0 &  0  &  0  & 0  & 0  \\
0 &  0  &  0  & 0  & 0  \\
0 &  0  &  0  & 0  & 0  \\
\end{array} \right] \label{example_Kij1}
\end{align} 
which means that only the topics $1$ and $2$ are coupled. But, the above matrix is not positive semidefinite. Thus, the basic condition of \textit{Lemma~\ref{lemma_nullspace_L}} is not satisfied. Let us consider another coupling matrix as
\begin{align}
K^{i,j} &=\left[\begin{array}{ccccc}   
1 &  1  &  0  & 0  & 0  \\
1 &  1  &  0  & 0  & 0  \\
0 &  0  &  0  & 0  & 0  \\
0 &  0  &  0  & 1  & 1  \\
0 &  0  &  0  & 1  & 1  \\
\end{array} \right] \label{example_Kij2}
\end{align} 
which is positive semidefinite. It follows that
\begin{align}
K^{i,j}\bm{\sigma}(x_j  - x_i)&=\left[\begin{array}{ccccc}   
1 &  1  &  0  & 0  & 0  \\
1 &  1  &  0  & 0  & 0  \\
0 &  0  &  0  & 0  & 0  \\
0 &  0  &  0  & 1  & 1  \\
0 &  0  &  0  & 1  & 1  \\
\end{array} \right]  \left(\begin{array}{c}   
\sigma( x_{j,1} - x_{i,1}) \\
\sigma( x_{j,2} - x_{i,2}) \\
\sigma( x_{j,3} - x_{i,3}) \\
\sigma( x_{j,4} - x_{i,4}) \\
\sigma( x_{j,5} - x_{i,5}) \\
\end{array} \right)\nonumber\\
&=  \left(\begin{array}{c}   
\sigma( x_{j,1} - x_{i,1}) + \sigma(x_{j,2} - x_{i,2})\\
\sigma(x_{j,1} - x_{i,1}) + \sigma( x_{j,2} - x_{i,2} ) \\
0  \\
\sigma( x_{j,4} - x_{i,4}) + \sigma(x_{j,5} - x_{i,5} )  \\
\sigma( x_{j,4} - x_{i,4}) + \sigma( x_{j,5} - x_{i,5}) \\
\end{array} \right) \nonumber
\end{align} 
Consequently, to satisfy $K^{i,j}\bm{\sigma}(x_j  - x_i)=0$, we need to have $x_{j,1} = x_{i,1}$, $x_{j,2} = x_{i,2}$, $x_{j,4} = x_{i,4}$, and $x_{j,5} = x_{i,5}$; but, $x_{j,3}$ and $x_{i,3}$ can be chosen  arbitrarily. 
\end{example}

From the above example, we can observe that the coupling matrices $K^{i,j},~ (i,j) \in \mathcal{E}$ would provide all possible consensus solutions in the topics among agents. Let us add another coupling between topics $1$ and $3$ into $K^{i,j}$ in \eqref{example_Kij2} as:
\begin{align}
K^{i,j} &=\left[\begin{array}{ccccc}   
1 &  1  &  1  & 0  & 0  \\
1 &  1  &  0  & 0  & 0  \\
1 &  0  &  0  & 0  & 0  \\
0 &  0  &  0  & 1  & 1  \\
0 &  0  &  0  & 1  & 1  \\
\end{array} \right] \label{example_Kij3}
\end{align} 
Then, against from an intuition, the matrix $K^{i,j}$ is no more positive semidefinite (actually it is indefinite since it has negative, zero, and positive eigenvalues). As there are more couplings between topics, the Laplacian matrix $\mathbb{L}$ has lost the positive semidefinite property. Thus, as far as the coupling matrix $K^{i,j}$ is positive semidefinite, the null space of $K^{i,j}$ enforces $x_j$ and $x_i$ to be synchronized. That is, in the multiplication $K^{i,j} \bm{\sigma}( x_j  - x_i)$, if the term $\sigma(x_{j,p}-x_{i,p})$ appears, then they will be synchronized; otherwise, if it does not appear, the synchronization $x_{j,p}  \rightarrow x_{i,p}$ is not enforced. Now, we can summarize the discussions as follows.
\begin{lemma} \label{lemma_neighbor_consensus} Let us assume that the Laplacian $\mathbb{L}(x)$ is positive semidefinite, and for two neighboring agents $i$ and $j$, the topics $p$ and $q$ are coupled, i.e., $K_{p,q}^{i,j} \neq 0$. Then, the opinion values $x_{i,p}$ and $x_{j,p}$ will reach a consensus and the opinion values $x_{i,q}$ and $x_{j,q}$ will also reach a consensus. If the same topic $p$ is coupled between neighboring agents, i.e., $K_{p,p}^{i,j} \neq 0$, then $x_{i,p}$ and $x_{j,p}$ will reach a consensus.
\end{lemma}

Given a  coupling matrix $K^{i,j}$, let us define a consensus matrix, ${C}^{i,j} =[ c_{p,q}^{i,j}] $, between agents $i$ and $j$ as
\begin{align}
c_{p,p}^{i,j} = c_{q,q}^{i,j} = \begin{cases}
1  &  \text{if}~K_{p,q}^{i,j}=1 \\
0  &  \text{if}~K_{p,q}^{i,j}=0 \\
\end{cases}
\end{align}
Now, we define the topic consensus graph $\mathcal{G}_{p,con} = (\mathcal{T}_{p,con}, \mathcal{E}_{p,con})$ for the topic $p$ as follows:
\begin{align}
\mathcal{T}_{p,con} & \triangleq \{ x_{1,p},  x_{2,p}, \ldots, x_{n,p} \} \\
\mathcal{E}_{p,con} & \triangleq \{ (i,j) ~|~ \exists (i,j)~\text{if}~c_{p,p}^{i,j}=1; ~\nexists (i,j) \nonumber\\
             &~~~~~~ \text{otherwise if}~c_{p,p}^{i,j}=0 \} 
\end{align}
So, $x_{j,p}$ is a neighbor of   $x_{i,p}$ if and only if $(i,j) \in \mathcal{E}_{p,con}$. That is, $x_{j,p} \in \mathcal{N}_{x_{i,p}}$ if and only if $c_{p,p}^{i,j} =1$. 
With the above definition, we can make the following main theorem.
\begin{theorem}  \label{thm_iff_consensus} Let us assume that the Laplacian $\mathbb{L}(x)$ is positive semidefinite. Then, for the topic $p$, the agents in $\mathcal{T}_{p,con}$ will reach a consensus if and only if the topic consensus graph $\mathcal{G}_{p,con}$ is connected. 
\end{theorem}
\begin{proof}
If the consensus graph $\mathcal{G}_{p,con}$ is connected, it can be considered that there is at least one path between $x_{i,p}$ and $x_{j,p}$ for any pair of $i$ and $j$. Thus, the set of $p$-topic agents, i.e., $\mathcal{T}_{p,con} \triangleq \{x_{1, p}, x_{2, p}, \ldots, x_{n,p} \}$, is connected, and the $p$-topic will reach a consensus by \textit{Lemma~\ref{lemma_neighbor_consensus}}. The \textit{only if} is also direct from   \textit{Lemma~\ref{lemma_neighbor_consensus}}. That is, if the topic opinions in $\mathcal{T}_{p,con}$ are not connected, it means that there are clusters, which are not connected. Thus, the consensus of the agents in $\mathcal{T}_{p,con}$ is not possible. 
\end{proof}

\begin{example}
Let us consider Fig.~\ref{consensus_graph_ex} that illustrates the interaction topology of a network. In this case, the elements of coupling matrix are given as
$K_{1,1}^{1,2} = K_{1,1}^{3,4} = K_{1,2}^{2,3} = K_{2,1}^{2,3} = K_{2,2}^{1,2} = K_{3,3}^{1,2} = K_{2,3}^{3,4} = K_{3,2}^{3,4} =1$. 
Thus, we have $c_{1,1}^{1,2} = c_{1,1}^{3,4} = c_{1,1}^{2,3} = c_{2,2}^{2,3} = c_{2,2}^{1,2} = c_{3,3}^{1,2} = c_{2,2}^{3,4} =c_{3,3}^{3,4} = 1$. Then, from this consensus matrix,  we can obtain the consensus graphs $\mathcal{G}_{p,con} = (\mathcal{T}_{p,con}, \mathcal{E}_{p,con})$ for $p=1, 2, 3$ with the following edge sets:
\begin{align}
\mathcal{E}_{1,con} & \triangleq \{ (1,2), (3,4), (2,3) \} \nonumber\\
\mathcal{E}_{2,con} & \triangleq \{ (2,3), (1,2), (3,4) \} \nonumber\\
\mathcal{E}_{3,con} & \triangleq \{ (1,2), (3,4) \} \nonumber
\end{align}
Therefore, the consensus graphs $\mathcal{G}_{1,con}$ and  $\mathcal{G}_{2,con}$ are connected, while the consensus graph $\mathcal{G}_{3,con}$ is not connected. 
\begin{figure}
\centering
\begin{tikzpicture}[scale=0.65]
\node[place, black] (node11) at (-3.6,3) [label=left:$$] {};
\node[place, black] (node12) at (-3.6,2) [label=left:$$] {};
\node[place, black] (node13) at (-3.6,1) [label=left:$$] {};

\node[place, black] (node21) at (-1.2,3) [label=left:$$] {};
\node[place, black] (node22) at (-1.2,2) [label=left:$$] {};
\node[place, black] (node23) at (-1.2,1) [label=left:$$] {};

\node[place, black] (node31) at (1.2,3) [label=left:$$] {};
\node[place, black] (node32) at (1.2,2) [label=left:$$] {};
\node[place, black] (node33) at (1.2,1) [label=left:$$] {};

\node[place, black] (node41) at (3.6,3) [label=left:$$] {};
\node[place, black] (node42) at (3.6,2) [label=left:$$] {};
\node[place, black] (node43) at (3.6,1) [label=left:$$] {};

\draw[black, line width=.5pt] (-3.6,2) ellipse (0.6cm and 2cm) ;
\draw[black, line width=.5pt] (-1.2,2) ellipse (0.6cm and 2cm);
\draw[black, line width=.5pt] (1.2,2) ellipse (0.6cm and 2cm);
\draw[black, line width=.5pt] (3.6,2) ellipse (0.6cm and 2cm);

\node[black] at (-3.6, 4.5) {Agent 1};
\node[black] at (-1.2, 4.5) {Agent 2};
\node[black] at (1.2, 4.5) {Agent 3};
\node[black] at (3.6, 4.5) {Agent 4};

\node[black] at (-5.6, 3) {Topic 1};
\node[black] at (-5.6, 2) {Topic 2};
\node[black] at (-5.6, 1) {Topic 3};


\draw (node11) [line width=0.5pt] -- node [left] {} (node21);
\draw (node12) [line width=0.5pt] -- node [right] {} (node22);
\draw (node13) [line width=0.5pt] -- node [below] {} (node23);

\draw (node21) [line width=0.5pt] -- node [left] {} (node32);
\draw (node22) [line width=0.5pt] -- node [right] {} (node31);


\draw (node31) [line width=0.5pt] -- node [left] {} (node41);
\draw (node32) [line width=0.5pt] -- node [right] {} (node43);
\draw (node33) [line width=0.5pt] -- node [below] {} (node42);

\end{tikzpicture}
\caption{Interaction topology of a network and coupling between neighboring topics.}
\label{consensus_graph_ex}
\end{figure}
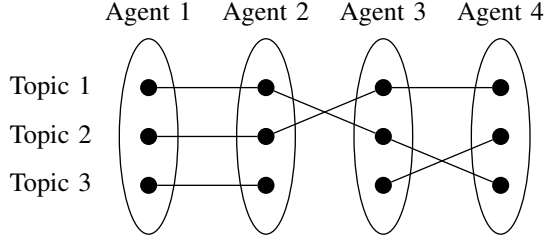
\end{example}

Now, by virtue of \textit{Theorem~\ref{thm_iff_consensus}}, we can conclude that if the Laplacian $\mathbb{L}(x)$ is positive semidefinite and all the topics are connected in the sense of \textit{Theorem~\ref{thm_iff_consensus}} (i.e., from the topic consensus graph $\mathcal{G}_{p,con} = (\mathcal{T}_{p,con}, \mathcal{E}_{p,con})$, then a complete consensus will be ensured. Otherwise, given $\mathcal{G}_{p,con}$, although the Laplacian $\mathbb{L}(x)$ is positive semidefinite, if the topic $p$ is not connected, then a partial opinion consensus will be achieved. The number of partial opinion clusters will be dependent on the number of clusters on the topic $p$. For example, in Fig.~\ref{partial_complete_consensus}, the topic $p=1$ has two clusters, the topic $p=2$ has two clusters, and the topic $p=3$ has one cluster. Let us define \textit{disconnection} as follows.
\begin{definition}
For a topic $p$, we call there is no disconnection if and only if the opinion values $x_{1,p}, x_{2,p}, \ldots, x_{n,p}$ are connected. If the opinion values are divided into $c_p$ components (there is no connection between components), then there are $c_p -1$ disconnections. 
\end{definition}
Then, for Fig.~\ref{partial_complete_consensus}, we can say that the topic $p=1$ has one disconnection (i.e., between agents $3$ and $4$), the topic $p=1$ also has one disconnection (i.e., between agents $1$ and $2$), and the topic $p=3$ does not have a disconnection. With the above definition, although it looks trivial, we can obtain the following observation. 
\begin{observation} \label{thm_psd_clusters}
Let there be $d$ topics, and each topic has $c_i, ~i=1, \ldots, d,$ clusters. Then there are $T_c = \sum_{k=1}^d (c_i -1) +1$ partial opinion clusters at maximum. 
\end{observation}
\begin{proof} Suppose that for the topic $d=1$, we have $c_1$ clusters. It means that there are $c_1 -1$ disconnections in the set $\mathcal{T}_{1,con} \triangleq \{x_{1, 1}, x_{2, 1}, \ldots, x_{n,1} \}$. Similarly, for the topic $d=2$ with $c_2$ clusters, there  are $c_2 -1$ disconnections in the set $\mathcal{T}_{2,con}$. Thus, by combining the topics $d=1$ and $d=2$, there could be at maximum $(c_1 -1) +(c_2 -1)$ disconnections in $\mathcal{T}_{1,con}$ and $\mathcal{T}_{2,con}$. Thus, if we consider all the topics,  there are at maximum $(c_1 -1) +(c_2 -1)$ disconnections in $\sum_{k=1}^d (c_i -1)$, which implies that there could be $T_c = \sum_{k=1}^d (c_i -1) +1$ partial opinion clusters at maximum. 
\end{proof}

The results thus far are developed for the inverse-proportional feedbacks. For the proportional feedbacks, we use \eqref{proportional_coupling_dynamics2}. The weighting matrix can be decomposed as \eqref{weighting_matrix_decompose}, with the diagonal matrix $\text{diag}(S^{i,j})$ given as:
\begin{align}
\text{diag}(S^{i,j}) = \text{diag}\Big(\frac{\text{sign}(x_{j,k} - x_{i,k})}{  c_{1} \Vert x_{j,k} - x_{i,k} \Vert + c_{0} }\Big) \label{inverse-proportional_diag_mat}
\end{align}
Also, $x^T \mathbb{L} x$ can be expressed as follows:
\begin{align}
x^T \mathbb{L} x &= \sum_{(i,j) \in \mathcal{E}}\bm{\eta}( x_j  - x_i)^T K^{i,j}\bm{\eta}(x_j  - x_i),\label{xTLx_null_proportional}
\end{align}        
where $\bm{\eta}(x_j  - x_i) \triangleq \big(\frac{{\sigma}(x_{j,1}  - x_{i,1})}{ c_{1} \Vert x_{j,1} - x_{i,1} \Vert + c_{0}  }, \frac{{\sigma}(x_{j,2}  - x_{i,2})}{c_{1} \Vert x_{j,2} - x_{i,2} \Vert + c_{0}},$  $\ldots, \frac{{\sigma}(x_{j,d}  - x_{i,d})}{c_{1} \Vert x_{j,d} - x_{i,d} \Vert + c_{0}  }\big)^T$. Thus, the null space of the Laplacian $\mathbb{L}$ in \eqref{xTLx_null_proportional} is same to the null space of $\mathbb{L}$ in \eqref{eq:xT L x_null}. Consequently, all the results in the inverse-proportional feedback couplings are exactly applied to the cases of the proportional feedback couplings.


\subsection{General Cases} \label{sec3_sub2_general}
The results in the previous section are quite clear and provide precise conditions for the characterization of opinion dynamics. However, the results are developed when the matrix $\text{blkdg}(K^{k\text{-th}})$ is positive semidefinite. As shown in \eqref{example_Kij3}, when the matrix $\text{blkdg}(K^{k\text{-th}})$ is not positive semidefinite, although it is against from intuition, there can be no theoretical guarantee for opinion consensus. For general case, we would like to directly analyze the stability of the inverse-proportional feedbacks modeled by \eqref{inverse_proportional_diagonal_dynamics} and \eqref{inverse_proportional_coupling_dynamics}. Let us take the Lyapunov candidate $V=\frac{1}{2} \Vert x \Vert^2$, which is radially unbounded and continuously differentiable, for the inverse-proportional feedbacks. The derivative of $V$ is computed as:
\begin{align} 
\dot{V}  &= -x^T \mathbb{L} x  \nonumber\\
&= -\sum_{(i,j) \in \mathcal{E} } \bm{\sigma}( x_j  - x_i)^T K^{i,j} \bm{\sigma}(x_j  - x_i) \nonumber\\
&=  -\underbrace{\sum_{(i,j) \in \mathcal{E} } \sum_{p=1}^d k_{p,p}^{i,j} ( \sigma(x_{j,p} - x_{i,p}))^2}_{ \triangleq \phi} \nonumber\\
&~~~~~ -\underbrace{\sum_{(i,j) \in \mathcal{E} } \sum_{p=1}^d \sum_{q=1,~q \neq p}^d k_{p,q}^{i,j} \sigma( x_{j,p} - x_{i,p}) \sigma( x_{j,q} - x_{i,q})}_{ \triangleq \psi }  \label{eq_v_dot_main}\\
&\leq 0 \nonumber
\end{align}
From the above inequality, it is clear that $\dot{V} =0$ if and only if $x_{j,p} = x_{i,p}$ for all topics, i.e., $\forall p \in \mathcal{T}$, the  opinion consensus is achieved, which is summarized as follows: 
\begin{theorem} \label{theorem_complete_consensus}
Let us suppose that the underlying interaction graph $\mathcal{G}$ is all-topic coupled, i.e, $\mathcal{G}_p$ is $p$-coupled for all $p \in \{1, \ldots, d\}$. Then, a complete opinion consensus is achieved. 
\end{theorem}
\begin{proof}
To make $\dot{V}=0$, it is required to have $\phi=0$ and $\psi=0$. Since $\mathcal{G}$ is all-topic coupled, $\phi=0$ implies  $\psi=0$. But, $\psi=0$ does not imply $\phi=0$. 
Thus, it is true that $\dot{V}=0$ if $x_{j,p} = x_{i,p}$ for all topics and for all $(i,j) \in \mathcal{E}$.
Suppose that there exists an edge such that $x_{j,p} \neq x_{i,p}$ for a specific topic $p$. Then, $\dot{V}\neq 0$. Thus, $\dot{V}=0$ only if $x_{j,p} = x_{i,p}$ for all topics and for all $(i,j) \in \mathcal{E}$. Consequently, the set $\mathcal{D} = \{x: x_{j,p} = x_{i,p},~\forall i, j \in \mathcal{V},~\forall p \in \mathcal{T}\}$ is the largest invariant set. Finally, by the Barbalat's lemma (due to $\ddot{V}$ exists and is bounded), the proof is completed. 
\end{proof}

%

\begin{remark} \label{remark_all_topic}
It is remarkable that the above results are true for both the homogeneous-coupling and heterogeneous-coupling networks, as far as the interaction graph $\mathcal{G}$ is all-topic coupled. 
\end{remark}

\begin{remark} \label{remark_not_p_connected_consensus}
It is noticeable that the condition of \textit{Theorem~\ref{theorem_complete_consensus}} is only a sufficient condition for a complete consensus. Thus, we may be able to achieve a complete opinion consensus even if the network is not all-topic coupled. 
Let us suppose that two topics $\bar{p}$ and $\bar{q}$ are not $p$-coupled. For example, the two topics are not directly coupled at the edge $(\bar{i},\bar{j})$. Since the overall network is connected, there must be terms such as $k_{\bar{p}, \bar{q}}^{\bar{i}, \bar{j}} \sigma( x_{\bar{j},\bar{p}} - x_{\bar{i},\bar{p}} ) \sigma( x_{\bar{j},\bar{q}} - x_{\bar{i},\bar{q}})$  in $\psi$. Thus, to make $\dot{V}=0$, it is required to have either $\sigma( x_{\bar{j},\bar{p}} - x_{\bar{i},\bar{p}} )=0$ or $\sigma( x_{\bar{j},\bar{q}} - x_{\bar{i},\bar{q}} )=0$. Therefore, even if the two topics $\bar{p}$ and $\bar{q}$ are not $p$-coupled, the neighboring agents $\bar{i}$ and $\bar{j}$ may reach a consensus. We will illustrate this case by an example in the simulation section. \end{remark}


From the equation \eqref{eq_v_dot_main}, we can see that 
if there is no cross couplings, i.e., $\psi =0$, then it is a usual consensus protocol in different layers. 
On the other hand, if there is no direct coupling, i.e., $\phi=0$, then there is no coupling in the same topics among agents. In the case of  $\psi =0$ with $k_{p,q}^{i,j}=0$ whenever $p \neq q$, it is still true that $\dot{V} =0$ if and only if $x_{j,p} = x_{i,p}$; thus, the typical consensus is achieved. Let $\phi=0$, with $k_{p,p}^{i,j}=0$ for all $p$. 
There are some undesired equilibrium cases. For example, given a coupling graph $\mathcal{G}_{i,j}$, let there exist paths from the topic node $1$ to all other topic nodes. That is, the graph $\mathcal{G}_{i,j}$ is a star graph with root node $1$. Then,  $\dot{V}$, with $\phi=0$, can be changed as:
\begin{align}
\dot{V} &= -\sum_{(i,j) \in \mathcal{E} }  \sigma( x_{j,1} - x_{i,1} ) \left[ \sum_{q=2}^d  k_{1,q}^{i,j} \sigma( x_{j,q} - x_{i,q} ) \right] 
\end{align}
So, if $\text{sgmd}( x_{j,1} - x_{i,1}) =0$ for all $(i,j) \in  \mathcal{E}$, then we have  $\dot{V} =0$. Thus, for a star graph, if the root topic has reached a consensus, all other topics may not reach a consensus. Actually, when $\phi=0$, a complete consensus is not achieved, due to the following reason:
\begin{claim} \label{coro_phi0}
Let us suppose that, $\forall (i,j) \in \mathcal{E}$, $k_{p,p}^{i,j} =0,~ \forall p$. 
Then, $\dot{V}$ will be almost zero (for the meaning of ``almost'', see the footnote $1$) with at least one topic having $x_{j,p} \neq x_{i,p}$ if and only if the coupling graphs $\mathcal{G}_{i,j},~\forall (i,j) \in \mathcal{E}$, are complete graphs. 
\end{claim}
\begin{proof} (\textit{If}) When it is a complete graph, without loss of generality, let the first topic, $p=1$, be reached a consensus. Then, we need to have 
$\sum_{p=2}^d \sum_{q=2,~q \neq p}^d k_{p,q}^{i,j} \sigma(x_{j,p} - x_{i,p}) \sigma(x_{j,q} - x_{i,q})=0$ to make $\dot{V}=0$. Similarly, suppose that the second topic has been reached a consensus, i.e., $p=2$. Then, we need to have $\sum_{p=3}^d \sum_{q=3,~q \neq p}^d k_{p,q}^{i,j} \sigma( x_{j,p} - x_{i,p} ) \sigma( x_{j,q} - x_{i,q} )=0$. By induction, when $p=d-1$, we need to have $k_{d,d-1}^{i,j} \sigma( x_{j,d} - x_{i,d} ) \sigma( x_{j,d-1} - x_{i,d-1} )=0$. So, to make $k_{d,d-1}^{i,j} \sigma( x_{j,d} - x_{i,d} ) \sigma( x_{j,d-1} - x_{i,d-1} )=0$, either $\sigma( x_{j,d} - x_{i,d} )$ or $\sigma( x_{j,d-1} - x_{i,d-1} )$ needs to be zero.\footnote{Actually, this does not imply that only one equality holds; the two equalities may hold. Indeed, suppose that, at time $t$, all other conditions have been satisfied, and $i$ and $j$ are still updating their opinions on topics $d$ and $d-1$ using the couplings $\sigma( x_{j,d}(t)-x_{i,d}(t))$ and $\sigma( x_{j,d-1}(t)-x_{i,-1d}(t))$. If those two coupling gains are equal, then the consensus speeds of $i$ and $j$ on topics $d$ and $d-1$  are equal. Thus, the topics $d$ and $d-1$ might achieve a consensus simultaneously. Although this case might rarely happen, it may occur; that is why we call it ``\textit{almost} zero with at least one topic having $x_{j,p} \neq x_{i,p}$''.} Thus, at least one topic does not need to reach a consensus. \\
(\textit{Only if}) Without loss of generality, let us suppose that there is no couple between the topics $p=d-1$ and $p=d$; but there are couplings between all other remaining topics. Then, by following the above procedure (``\textit{if}'' procedure), when $p=d-2$, we have  $k_{d-1,d-2}^{i,j} \sigma( x_{j,d-1} - x_{i,d-1} ) \sigma( x_{j,d-2} - x_{i,d-2} ) + k_{d,d-2}^{i,j} \sigma( x_{j,d} - x_{i,d} ) \sigma( x_{j,d-2} - x_{i,d-2} ) =0$. Thus, if it is assumed that $\sigma( x_{j,d-2} - x_{i,d-2} ) =0$, then the two topics, $p=d-1$ and $p=d$, do not need to reach a consensus.  
\end{proof}

\begin{remark}
In \textit{Claim~\ref{coro_phi0}}, since $k_{p,p}^{i,j}=0$ for all $p$ and for all edges, and the coupling graphs are complete graphs, it can be classified as a homogeneous-coupling network. 
\end{remark}

The above claim implies that a complete opinion consensus for all topics is not ensured for general graphs, when $\phi=0$. 
Also under the condition of $\phi=0$, when the coupling graphs are not complete graphs, it is likely that more than one topics would not reach consensus. 
Thus, for a complete opinion consensus, it is required to have $\phi \neq 0$.  
\begin{observation}\label{theorem_phi0_partial}
Consider a homogeneous-coupling network. Let $\phi \neq 0$; but $k_{p,p}^{i,j} =0$ for some $p \in \mathcal{T}$, $\forall (i,j) \in \mathcal{E}$. 
Then, a complete opinion consensus is not ensured. 
\end{observation}
\begin{proof}
Let us divide the set $\mathcal{T}$ as $\mathcal{T} = \mathcal{T}^\circ \cup \mathcal{T}^\times$ and $\mathcal{T}^\circ \cap \mathcal{T}^\times = \emptyset$, where $k_{p,p}^{i,j} \neq 0$ when $p \in \mathcal{T}^\circ$ and  $k_{p,p}^{i,j} = 0$ when $p \in \mathcal{T}^\times$.  Then, for all the topics $p \in \mathcal{T}^\circ$, we need to have $\sigma( x_{j,p}- x_{i,p} ) =0$ to make $\dot{V}=0$. Then, to make $\psi=0$, it is required to $k_{p,q}^{i,j} \sigma( x_{j,p} - x_{i,p} ) \sigma( x_{j,q} - x_{i,q} ) =0$ when $p \in \mathcal{T}^\circ$ and $q \in \mathcal{T}^\times$, or $k_{p,q}^{i,j} \sigma( x_{j,p} - x_{i,p} ) \sigma( x_{j,q} - x_{i,q} ) =0$ when $p, q \in \mathcal{T}^\times$. For the former case, since $\sigma( x_{j,p} - x_{i,p} ) =0$, it does not need to have $\sigma( x_{j,q} - x_{i,q} ) =0$. Thus, for the topics $q \in \mathcal{T}^\times$, a consensus may not be achieved. For the latter case, due to the same reason as the proof of \textit{Claim~\ref{coro_phi0}}, there will be some topics that do not reach a consensus. 
\end{proof}

 \textit{Theorem~\ref{theorem_complete_consensus}} and \textit{Observation~\ref{theorem_phi0_partial}} lead a conclusion that each  topic needs to be $p$-coupled to have a complete consensus. However, as remarked in \textit{Remark~\ref{remark_not_p_connected_consensus}}, it is not argued that the $p$-coupling for all topics, i.e., all-topic coupled, is the necessary and sufficient condition for a complete opinion consensus. 
From the equation \eqref{eq_v_dot_main}, we can infer that the interdependent couplings between topics are required to speed up the opinion consensus. So, to have an opinion consensus on a topic, the agents of the society need to discuss directly on the same topic. But, if they have some opinion couplings with other topics, the consensus of the topic may be achieved more quickly. 

Next, let us consider the proportional feedbacks modeled by \eqref{proportional_diagonal_dynamics} and  \eqref{proportional_coupling_dynamics2}.  For the proportional feedbacks, using the same Lyapunov candidate $V= \frac{1}{2}\Vert x \Vert^2$, we can obtain the derivative of $V$ as:
\begin{align} 
\dot{V} &= -\sum_{(i,j) \in \mathcal{E} } \sum_{p=1}^d \sum_{q=1}^d k_{p,q}^{i,j}  \nonumber\\
         &~~~~~ \times \frac{\sigma( x_{j,p} - x_{i,p} ) \sigma( x_{j,q} - x_{i,q} )}{{ (c_{1} \Vert x_{j,p} - x_{i,p} \Vert + c_{0}) (c_{1} \Vert x_{j,q} - x_{i,q} \Vert + c_{0})  }}  \leq 0   \label{eq_v_dot_main3}
\end{align}
Since the denominator of the right-hand side of  \eqref{eq_v_dot_main3} is always positive, the equilibrium set for $\dot{V}=0$ is decided if and only if $\sigma( x_{j,p} - x_{i,p} ) \sigma( x_{j,q} - x_{i,q} ) = 0$ for all $p, q \in \mathcal{T}$. Consequently, we have the same results as the inverse-proportional feedback couplings.


\begin{observation} \label{theorem_hetero_analysis}
Let us consider general heterogeneous-coupling network, i.e., $\mathcal{G}_{i_1,j_1}\neq \mathcal{G}_{i_2,j_2}$ for some edges $(i_1, j_1) \neq (i_2, j_2)$. If some topics are not $p$-coupled, then a complete opinion consensus is not ensured.
\end{observation} 
\begin{proof}
Let us suppose that there is no direct coupling between agents $\bar{j}$ and $\bar{i}$, on a specific topic $\bar{p}$. Then, in $\phi$ of \eqref{eq_v_dot_main}, the term $(x_{\bar{j}, \bar{p}}- x_{\bar{i}, \bar{p}})^2$ is missed. But, the term $\sigma( x_{\bar{j}, \bar{p}}- x_{\bar{i}, \bar{p}} )$ may be included in $\psi$ in the form of $\sigma( x_{\bar{j}, \bar{p}}- x_{\bar{i}, \bar{p}} )  \sigma( x_{\bar{j}, {p}}- x_{\bar{i}, {p}} )$ if there are cross couplings between the topic $\bar{p}$ and any other topics $p$. If there is a direct coupling on the topic $p$ between agents $\bar{j}$ and $\bar{i}$, then the term $\sigma( x_{\bar{j}, {p}}- x_{\bar{i}, {p}} )$ will be zero; thus, $\sigma( x_{\bar{j}, \bar{p}}- x_{\bar{i}, \bar{p}} )$ does not need to be zero to make $\dot{V}$ zero. Or, if there is no direct coupling on the topic $p$ between agents $\bar{j}$ and $\bar{i}$, still either $\sigma( x_{\bar{j}, \bar{p}}- x_{\bar{i}, \bar{p}} )$ or $\sigma( x_{\bar{j}, {p}}- x_{\bar{i}, {p}} )$ does not need to be zero also. Thus, a complete opinion consensus is not ensured. 
\end{proof}

The results of \textit{Observation~\ref{theorem_phi0_partial}} and \textit{Observation~\ref{theorem_hetero_analysis}} leave a question about the clustered opinions. Let us consider a network depicted in Fig.~\ref{four_agent_example}. From the term $\phi$ in \eqref{eq_v_dot_main}, all the topics between agents $1$ and $2$, and all the topics between agents $3$ and $4$ reach an opinion consensus. Due to the interdependent couplings between agents $2$ and $3$, we have the interdependency terms as 
$\psi = k_{1,2}^{2,3} \sigma( x_{2,1} - x_{3,1} ) \sigma( x_{2,2} - x_{3,2} ) + k_{2,1}^{3,2} \sigma( x_{2,2} - x_{3,2} ) \sigma( x_{2,1} - x_{3,1} )  + k_{2,3}^{2,3} \sigma( x_{2,2} - x_{3,2} ) \sigma( x_{2,3} - x_{3,3} ) + k_{3,2}^{3,2} \sigma( x_{2,3} - x_{3,3} ) \sigma( x_{2,2} - x_{3,2} )$. 
Thus, by Barbalat's lemma, to make $\dot{V}$ zero, we need to have $\psi =0$. From the above equation, for example, if $\sigma( x_{2,2} - x_{3,2} )=0$, then $\psi$ becomes zero. The largest invariant set for having $\dot{V}=0$ is obtained as $\mathcal{D} = \mathcal{D}^d \cup \mathcal{D}^u$, where the desired set is given
\begin{align}
\mathcal{D}^d = \{ x: x_1 = x_2 = x_3 = x_4 \} \nonumber
\end{align}
and undesired set is given as
\begin{align}
\mathcal{D}^u = \{ x: x_1 = x_2, x_3 = x_4, x_2 \neq x_3 \} \nonumber
\end{align}
In the undesired set, the opinions of agents $2$ and $3$ may be related as (i) $x_{2,2} = x_{3,2}$, but $x_{2,1} \neq x_{3,1}$ and $x_{2,3} \neq x_{3,3}$, (ii) $x_{2,2} \neq x_{3,2}$, but $x_{2,1} = x_{3,1}$ and $x_{2,3} = x_{3,3}$, (iii) $x_{2,3} \neq x_{3,3}$, but $x_{2,1} = x_{3,1}$ and $x_{2,2} = x_{3,2}$, or (iv) $x_{2,1} \neq x_{3,1}$, but $x_{2,2} = x_{3,2}$ and $x_{2,3} = x_{3,3}$. Thus, a part of opinions reaches a consensus, while a part of opinions may reach clustered consensus. 

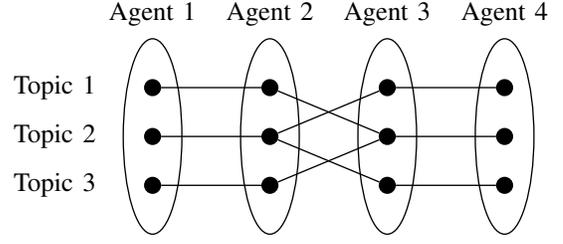
\begin{figure}
\centering
\begin{tikzpicture}[scale=0.65]
\node[place, black] (node11) at (-3.6,3) [label=left:$$] {};
\node[place, black] (node12) at (-3.6,2) [label=left:$$] {};
\node[place, black] (node13) at (-3.6,1) [label=left:$$] {};

\node[place, black] (node21) at (-1.2,3) [label=left:$$] {};
\node[place, black] (node22) at (-1.2,2) [label=left:$$] {};
\node[place, black] (node23) at (-1.2,1) [label=left:$$] {};

\node[place, black] (node31) at (1.2,3) [label=left:$$] {};
\node[place, black] (node32) at (1.2,2) [label=left:$$] {};
\node[place, black] (node33) at (1.2,1) [label=left:$$] {};

\node[place, black] (node41) at (3.6,3) [label=left:$$] {};
\node[place, black] (node42) at (3.6,2) [label=left:$$] {};
\node[place, black] (node43) at (3.6,1) [label=left:$$] {};

\draw[black, line width=.5pt] (-3.6,2) ellipse (0.6cm and 2cm) ;
\draw[black, line width=.5pt] (-1.2,2) ellipse (0.6cm and 2cm);
\draw[black, line width=.5pt] (1.2,2) ellipse (0.6cm and 2cm);
\draw[black, line width=.5pt] (3.6,2) ellipse (0.6cm and 2cm);

\node[black] at (-3.6, 4.5) {Agent 1};
\node[black] at (-1.2, 4.5) {Agent 2};
\node[black] at (1.2, 4.5) {Agent 3};
\node[black] at (3.6, 4.5) {Agent 4};

\node[black] at (-5.6, 3) {Topic 1};
\node[black] at (-5.6, 2) {Topic 2};
\node[black] at (-5.6, 1) {Topic 3};


\draw (node11) [line width=0.5pt] -- node [left] {} (node21);
\draw (node12) [line width=0.5pt] -- node [right] {} (node22);
\draw (node13) [line width=0.5pt] -- node [below] {} (node23);

\draw (node21) [line width=0.5pt] -- node [left] {} (node32);
\draw (node22) [line width=0.5pt] -- node [right] {} (node31);

\draw (node22) [line width=0.5pt] -- node [left] {} (node33);
\draw (node23) [line width=0.5pt] -- node [right] {} (node32);

\draw (node31) [line width=0.5pt] -- node [left] {} (node41);
\draw (node32) [line width=0.5pt] -- node [right] {} (node42);
\draw (node33) [line width=0.5pt] -- node [below] {} (node43);

\end{tikzpicture}
\caption{A network composed of four agents with three topics.}
\label{four_agent_example}
\end{figure}

It is clear that if there are some topics that are $p$-coupled, then a complete clustered consensus cannot take place. Also, even though the network is not $p$-coupled for all $p$, if the network is connected, then a complete clustered consensus is not ensured since the connected neighboring topics would reach a consensus.  Thus, a complete opinion consensus rarely occurs as far as the network is connected. But, a partial opinion consensus would occur easily if it is not all-topic coupled. In fact, if the network is not all-topic coupled, the network would have opinion-based clustered consensus. 
It means that if agents of network are connected, some opinions would be agreed among agents, but some opinions would be divided into clusters. Or, most of opinions would be clustered, depending on interaction network topology $\mathcal{G}$ and the topic topologies $\mathcal{G}_p$.  
\begin{observation}\label{thm_complete_opinion_consensus_not} Suppose that a network is connected. Even though $\phi=0$, a complete clustered consensus is not ensured.
\end{observation}
\begin{proof}
Due to the term $\psi$ including $\sigma( x_{j,p} - x_{i,p} ) \sigma( x_{j,q} - x_{i,q} )$, at least one of the topics $p$ and $q$ needs to be agreed. Thus, a complete clustered opinion consensus does not occur.
\end{proof}

Now, with the statements of \textit{Observation~\ref{theorem_phi0_partial}}, \textit{Observation~\ref{theorem_hetero_analysis}} and \textit{Observation~\ref{thm_complete_opinion_consensus_not}}, we can see that if agents of a society are not all-topic coupled, but just connected in the sense of interaction graph $\mathcal{G}$, then both a complete opinion consensus and complete clustered consensus are not ensured. 





\section{Simulations} \label{section_sim}
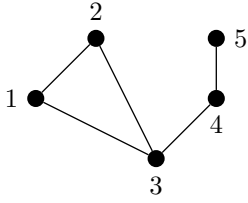
\begin{figure}
\centering
\begin{tikzpicture}[scale=0.8]
\node[place, black] (node1) at (-2,0) [label=left:$1$] {};
\node[place, black] (node2) at (-1,1) [label=above:$2$] {};
\node[place, black] (node3) at (0,-1) [label=below:$3$] {};
\node[place, black] (node4) at (1,0) [label=below:$4$] {};
\node[place, black] (node5) at (1,1) [label=right:$5$] {};


\draw (node1) [line width=0.5pt] -- node [left] {} (node2);
\draw (node1) [line width=0.5pt] -- node [right] {} (node3);
\draw (node2) [line width=0.5pt] -- node [left] {} (node3);
\draw (node3) [line width=0.5pt] -- node [right] {} (node4);
\draw (node4) [line width=0.5pt] -- node [left] {} (node5);
\end{tikzpicture}
\caption{Underlying topology for numerical simulations.}
\label{example1_five_agents}
\end{figure}

\subsection{Case of Positive Semidefinite Laplacian}
Let us consider five agents with the underlying interaction network topology as depicted in Fig.~\ref{example1_five_agents}. The initial opinions of agents are given as $x_1 =(1, 2, 3)^T$, $x_2 =(2, 4, 4)^T$, $x_3 =(3, 1, 5)^T$, $x_4 =(4, 3, 2)^T$, $x_5 =(5, 6, 1)^T$. The initial opinions of agents for the three topics are different each other. To verify the results of Section~\ref{sec3_sub1_psd}, the following coupling matrices are considered. 
\begin{align}
K^{1,2} &=\left[\begin{array}{ccc}   
1 & 1   & 0  \\
1 & 1   & 0  \\
0 & 0   & 0  \\
\end{array} \right]; K^{1,3} =\left[\begin{array}{ccc}   
1 & 0   & 0  \\
0 & 1   & 1  \\
0 & 1   & 1  \\
\end{array} \right] \nonumber\\
 K^{2,3} &=\left[\begin{array}{ccc}   
2 & 0   & 1  \\
0 & 2   & 1  \\
1 & 1   & 2  \\
\end{array} \right]; K^{3,4} =\left[\begin{array}{ccc}   
1 & 1   & 1  \\
1 & 1   & 1  \\
1 & 1   & 1  \\
\end{array} \right]\nonumber\\
 K^{4,5} &=\left[\begin{array}{ccc}   
1 & 0   & 1  \\
0 & 1   & 0  \\
1 & 0   & 1  \\
\end{array} \right] \label{coupling_mat_sim_1}
\end{align}
which are all positive semidefinite. From the above coupling matrices, it is shown that the topic consensus graph $\mathcal{G}_{p,con}$ for all $p=1, 2, 3$ is connected. Thus, as expected from \textit{Theorem~\ref{thm_iff_consensus}}, a consensus for all topics is achieved. Fig.~\ref{sim_psd_consensus} shows that all the values of the topics of agents reach a consensus as time passes. 
\begin{figure*}
	\centering 
	\includegraphics[width=0.96\textwidth]{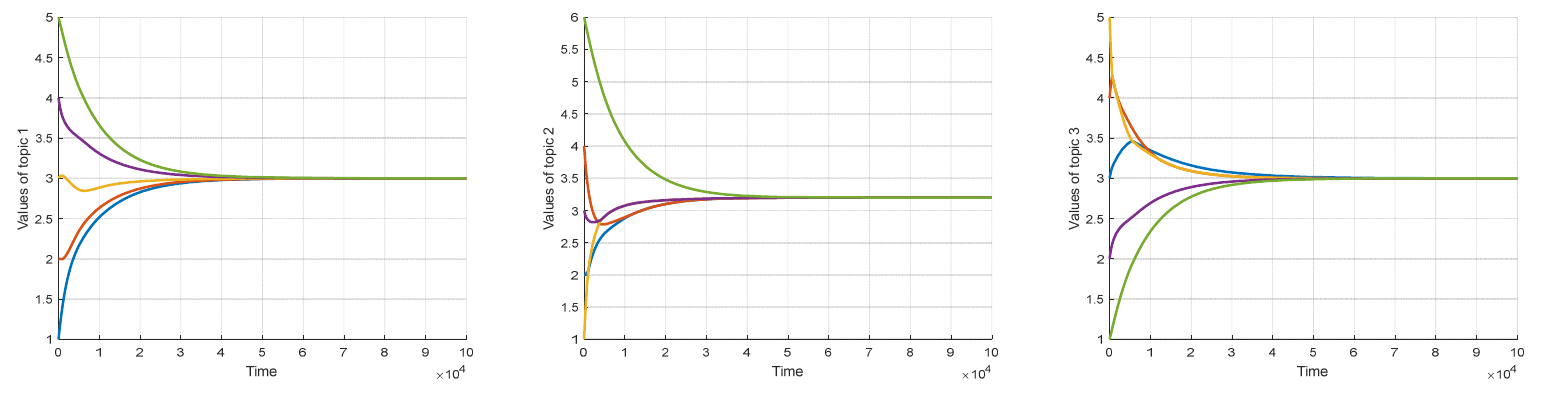}
	\caption{Consensus under positive semidefinite Laplacian: Left - Topic $1$ (i.e., $x_{i,1},~i=1,\ldots, 5$). Center - Topic $2$ (i.e., $x_{i,2},~i=1,\ldots, 5$). Right - Topic $3$  (i.e., $x_{i,3},~i=1,\ldots, 5$).}
	\label{sim_psd_consensus}
\end{figure*}
Next, let us change $K^{1,3}$ and $K^{3,4}$ as 
\begin{align}
K^{1,3} =\left[\begin{array}{ccc}   
0 & 0   & 0  \\
0 & 1   & 1  \\
0 & 1   & 1  \\
\end{array} \right]; K^{3,4} =\left[\begin{array}{ccc}   
0 & 0   & 0  \\
0 & 1   & 1  \\
0 & 1   & 1  \\
\end{array} \right] \nonumber
\end{align}
which are still positive semidefinite. However, due to the new $K^{3,4}$, there is a disconnection in topic $1$ between agents $3$ and $4$. Thus, the topic $1$ is not connected in the topic consensus graph $\mathcal{G}_{1,con}$. As expected from \textit{Theorem~\ref{thm_psd_clusters}}, there will be two clusters. Fig.~\ref{sim_example_psd_clusters} shows that the topic $1$ does not reach a consensus; there are two clusters (one cluster with agents $1, 2$, and $3$, and another cluster with agents $4$ and $5$).
\begin{figure*}
	\centering 
	\includegraphics[width=0.97\textwidth]{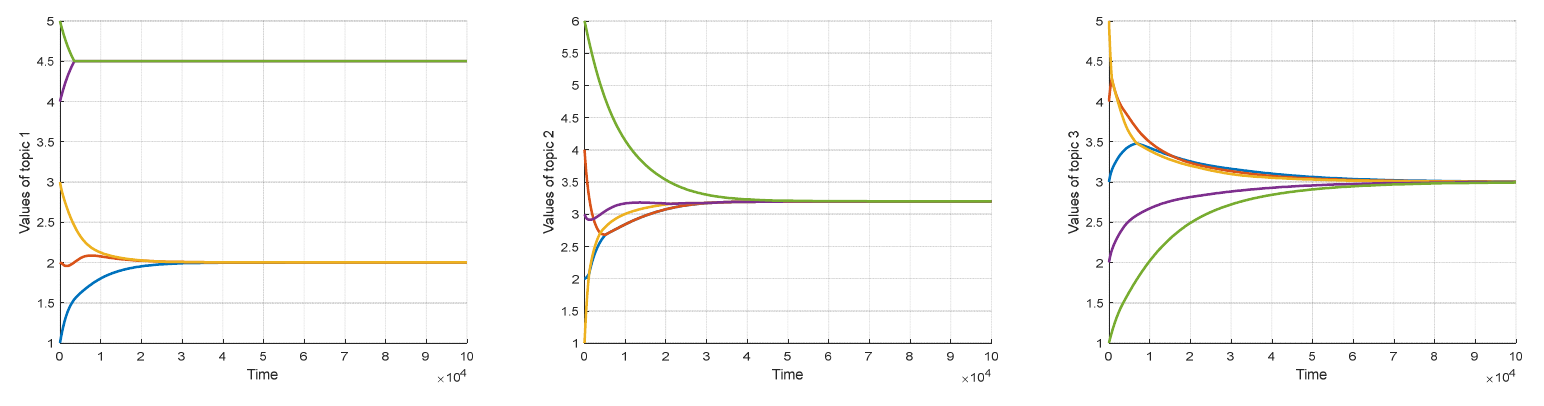}
	\caption{Partial opinion consensus with a disconnected topic consensus graph.}
	\label{sim_example_psd_clusters}
\end{figure*}

\subsection{General Cases}
Let the coupling topologies for each edge be given as:
\begin{align}
K^{1,2} &=\left[\begin{array}{ccc}   
1 & 1   & 0  \\
1 & 1   & 1  \\
0 & 1   & 1  \\
\end{array} \right]; K^{1,3} =\left[\begin{array}{ccc}   
1 & 1   & 0  \\
1 & 1   & 1  \\
0 & 1   & 1  \\
\end{array} \right]\nonumber\\
K^{2,3} &=\left[\begin{array}{ccc}   
1 & 1   & 0  \\
1 & 1   & 1  \\
0 & 1   & 1  \\
\end{array} \right]; K^{3,4} =\left[\begin{array}{ccc}   
1 & 1   & 0  \\
1 & 1   & 1  \\
0 & 1   & 1  \\
\end{array} \right]\nonumber\\
K^{4,5} &=\left[\begin{array}{ccc}   
1 & 1   & 0  \\
1 & 1   & 1  \\
0 & 1   & 1  \\
\end{array} \right] \nonumber
\end{align}
which are indefinite matrices. Since all the topics are $p$-coupled, it is an all-topic coupled network. Also, since the coupling matrices for all edges are equivalent, it is a homogeneous network. With the above coupling matrices, as expected from \textit{Theorem~\ref{theorem_complete_consensus}}, the topics of agents reach a complete opinion consensus. Next, let us change the matrix $K^{3,4}$ as 
\begin{align}
K^{3,4} &=\left[\begin{array}{ccc}   
0 & 1   & 0  \\
1 & 0   & 1  \\
0 & 1   & 1  \\
\end{array} \right] \label{a34_not_consensus}
\end{align} 
In this case, the topic $1$ and $2$ are not $p$-coupled, although the underlying interaction network is connected. As observed in \textit{Observation~\ref{theorem_hetero_analysis}}, Fig.~\ref{sim_example2} shows that the topic $1$ does not reach a consensus, while the topic $2$ still reaches a consensus. In the topic $1$, agents $3$, $4$ and $5$ reach a consensus, while agents $4$ and $5$ reach a consensus. But, when the matrix $A_{3,4}$ is changed again as
\begin{align}
K^{3,4} &=\left[\begin{array}{ccc}   
1 & 1   & 0  \\
1 & 0   & 1  \\
0 & 1   & 1  \\
\end{array} \right] \label{a34_yes_consensus}
\end{align} 
all the topics have reached a consensus although it is not all-topic coupled. Let us change the weight matrices $K^{2,3}$ and $K^{1,3}$ as
\begin{align}
K^{1,3} =\left[\begin{array}{ccc}   
0 & 1   & 0  \\
1 & 1   & 1  \\
0 & 1   & 0  \\
\end{array} \right]; K^{2,3} =\left[\begin{array}{ccc}   
0 & 1   & 0  \\
1 & 1   & 1  \\
0 & 1   & 0  \\
\end{array} \right] 
\end{align} 
In this case, the network is not all-topic coupled. As shown in Fig.~\ref{sim_example3}, the topics $1$ and $3$ do not reach a consensus, while the topic $2$ reaches a consensus. 
Next, let us consider $\phi=0$ and $\mathcal{G}_{i,j} ~\forall (i,j) \in \mathcal{E}$ are complete graphs. Fig.~\ref{sim_example5} shows the simulation result. All the topics do not reach a consensus. 

\begin{figure*}
	\centering 
	\includegraphics[width=0.95\textwidth]{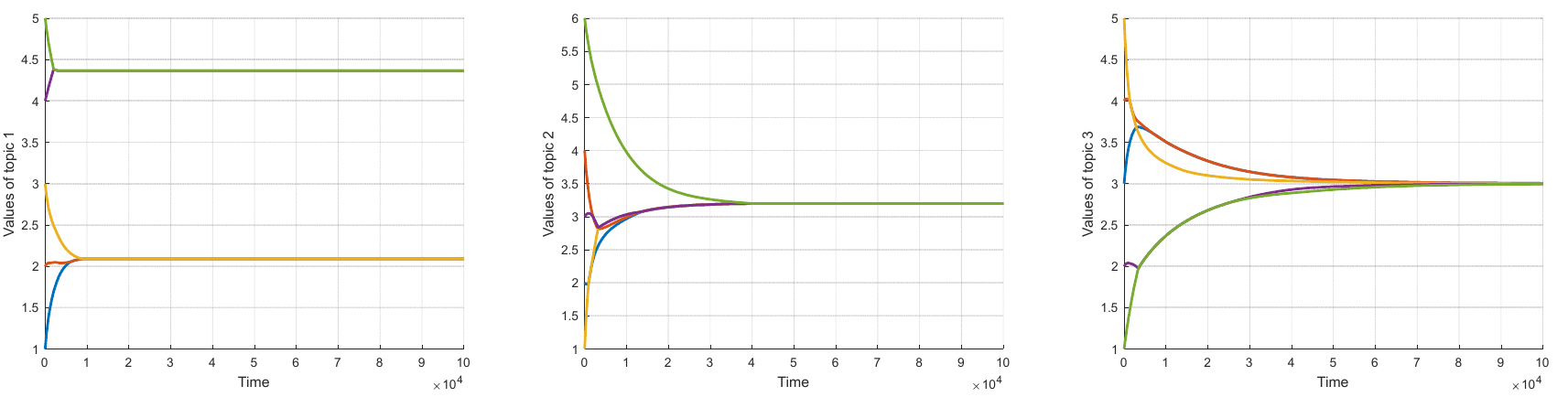}
	\caption{The topics $1$ and $2$ are not $p$-connected, due to zero diagonal terms in $K^{3,4}$: Left - Topic $1$ (i.e., $x_{i,1},~i=1,\ldots, 5$). The agents $4$ and $5$ reach a consensus, and agents $1, 2$ and $3$ reach a consensus for the topic $1$. Center - Topic $2$ (i.e., $x_{i,2},~i=1,\ldots, 5$). Right - Topic $3$  (i.e., $x_{i,3},~i=1,\ldots, 5$).}
	\label{sim_example2}
\end{figure*}
\begin{figure*}
	\centering 
	\includegraphics[width=0.95\textwidth]{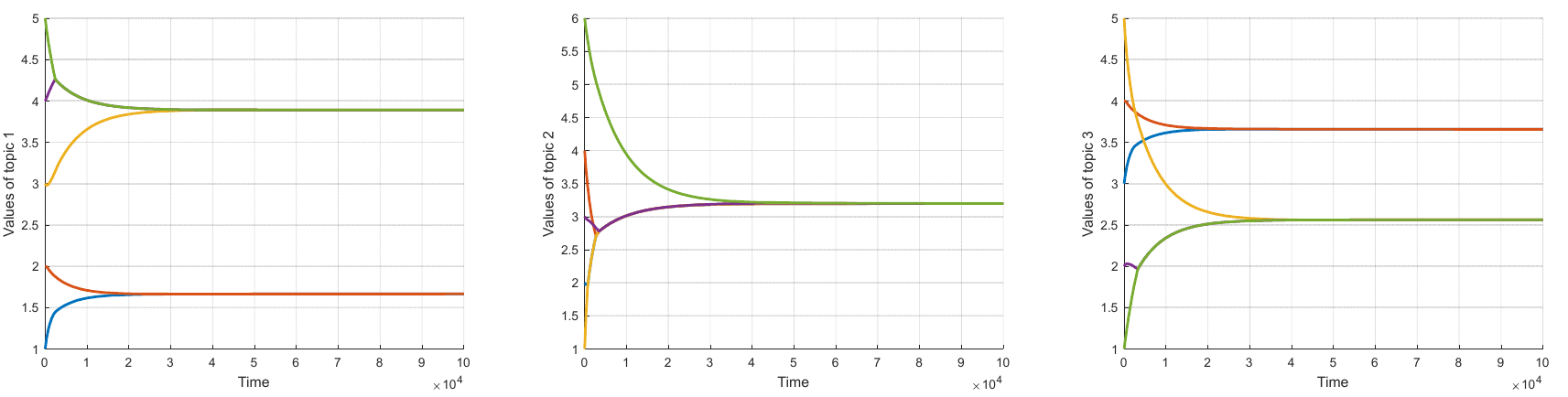}
	\caption{Not all-topic connected, with zero diagonal terms in $K^{1,3}$ and $K^{2,3}$; only the topic $2$ is $p$-connected. The topics $1$ and $3$ do not reach a consensus (clustered), while the topic $2$ reaches a consensus. For both the topics $1$ and $3$, agents $1$ and  $2$ reach a consensus, and agents $3, 4$, and $5$ reach a consensus.}
	\label{sim_example3}
\end{figure*}
\begin{figure*}
	\centering 
	\includegraphics[width=0.95\textwidth]{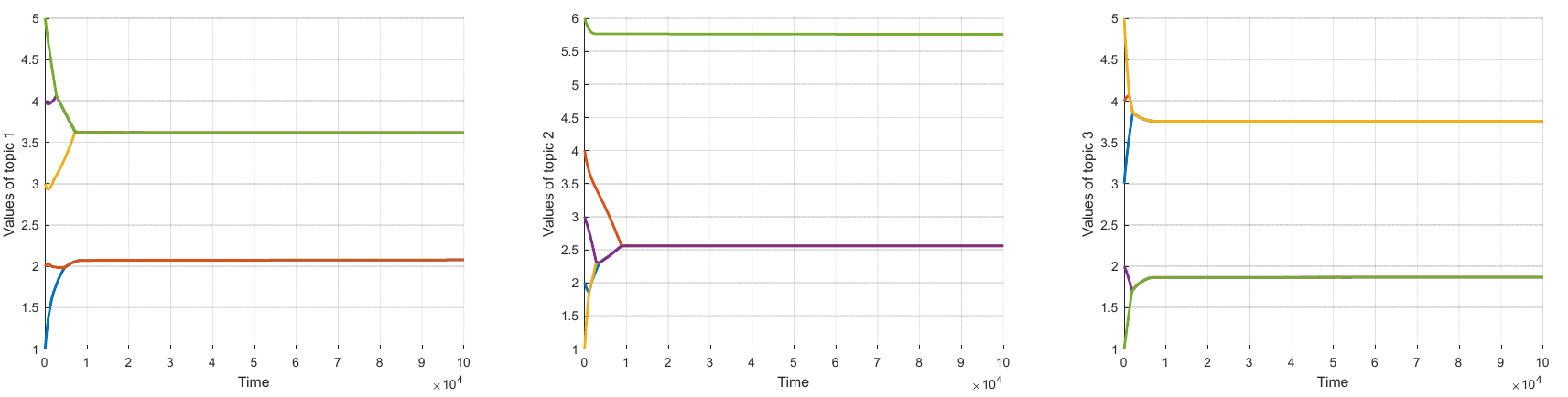}
	\caption{$\phi=0$ with complete interdependency graphs. The agents do not reach a consensus even for a topic.}
	\label{sim_example5}
\end{figure*}

\section{Conclusion} \label{section_con}
The cooperative opinion dynamics on multiple interdependent topics may be considered as a consensus problem of multi-layer networks. 
Each topic can be considered as a basic layer and the term $a^{i,j}_{p,q}$ may describe a cross-layer connection between the layer $p$ and layer $q$, and between agent $i$ and agent $j$. The basic layer is the direct connections that are essential for achieving a consensus on this layer. 
This paper shows that the opinion dynamics with multiple interdependent topics, which is the consensus dynamics in multi-layer networks, possesses some new properties different from the usual consensus in one layer. Clustering phenomenon occurs quite often, even though the number of connections between agents is large. In general, adding a direct connection $a^{i,j}_{p,p}$ forces a consensus between agents $i$ and $j$ on the topic $p$. On the other hand, adding a set of cross-layer connections $\{a^{i,j}_{p,q}\}_{q \neq p, q = 1, \ldots, d}$ may not so significantly helpful for the agents $i$ and $j$ to reach a consensus on topic $p$. But, from simulations, it is shown that the cross-layer connections are still beneficial for a consensus on the topics. Of course, as analyzed in the case of positive semidefinite Laplacian matrices, the cross couplings are also very helpful for a consensus. 
In our future efforts, we would like to evaluate the polarization phenomenon of bipartite graphs under the setup of multiple interdependent couplings, which may be a general one of \cite{Altafini_tac_2013,Hendrickx_cdc_2014} in multidimensional spaces.  It is also interesting to change the overall formulation in discrete-time cases; then the discontinuity arising in the sign functions can be handled more easily. We are also interested in the problem of switches in the coupling matrices (for example, a coupling matrix could switch from a positive semidefinite property to indefinite property). Then, the topology will be time-variant. In our future efforts, we would like to solve this problem in a more general setup.

\section*{Acknowledgment}
The work of this paper has been supported by the National Research Foundation (NRF) of Korea under the grant NRF-2017R1A2B3007034. The results of this paper have been developed on the basis of the observations given in \cite{Ahn_etal2018_opinion_arvix}.


\bibliographystyle{IEEEtran}
\bibliography{opinion_dynamics}

\begin{thebibliography}{10}
\providecommand{\url}[1]{#1}
\csname url@samestyle\endcsname
\providecommand{\newblock}{\relax}
\providecommand{\bibinfo}[2]{#2}
\providecommand{\BIBentrySTDinterwordspacing}{\spaceskip=0pt\relax}
\providecommand{\BIBentryALTinterwordstretchfactor}{4}
\providecommand{\BIBentryALTinterwordspacing}{\spaceskip=\fontdimen2\font plus
\BIBentryALTinterwordstretchfactor\fontdimen3\font minus
  \fontdimen4\font\relax}
\providecommand{\BIBforeignlanguage}[2]{{%
\expandafter\ifx\csname l@#1\endcsname\relax
\typeout{** WARNING: IEEEtran.bst: No hyphenation pattern has been}%
\typeout{** loaded for the language `#1'. Using the pattern for}%
\typeout{** the default language instead.}%
\else
\language=\csname l@#1\endcsname
\fi
#2}}
\providecommand{\BIBdecl}{\relax}
\BIBdecl

\bibitem{Jiang_etal_tsmcs_2018}
L.~Jiang, J.~Liu, D.~Zhou, Q.~Zhou, X.~Yang, and G.~Yu, ``Predicting the
  evolution of hot topics: A solution based on the online opinion dynamics
  model in social network,'' \emph{IEEE Transactions on Systems, Man, and
  Cybernetics: Systems}, pp. 1--13, 2018.

\bibitem{Cho_tcss_2018}
J.-H. Cho, ``Dynamics of uncertain and conflicting opinions in social
  networks,'' \emph{IEEE Transactions on Computational Social Systems}, vol.~5,
  no.~2, pp. 518--531, 2018.

\bibitem{Dietrich_etal_tac_2018}
F.~Dietrich, S.~Martin, and M.~Jungers, ``Control via leadership of opinion
  dynamics with state and time-dependent interactions,'' \emph{IEEE
  Transactions on Automatic Control}, vol.~63, no.~4, pp. 1200--1207, 2018.

\bibitem{Etesami_Sasar_tac_2015}
S.~R. Etesami and T.~Basar, ``Game-theoretic analysis of the
  {Hegselmann-Krause} model for opinion dynamics in finite dimensions,''
  \emph{IEEE Transactions on Automatic Control}, vol.~60, no.~7, pp.
  1886--1897, 2015.

\bibitem{Wedin_Hegarty_tac_2015}
E.~Wedin and P.~Hegarty, ``The {Hegselmann-Krause} dynamics for the
  continuous-agent model and a regular opinion function do not always lead to
  consensus,'' \emph{IEEE Transactions on Automatic Control}, vol.~60, no.~9,
  pp. 2416--2421, 2015.

\bibitem{Leshem_Scaglione_tsipn_2018}
A.~Leshem and A.~Scaglione, ``The impact of random actions on opinion
  dynamics,'' \emph{IEEE Transactions on Signal and Information Processing over
  Networks}, vol.~4, no.~3, pp. 576--584, 2018.

\bibitem{Semonsen_etal_tcyb_2018}
J.~Semonsen, C.~Griffin, A.~Squicciarini, and S.~Rajtmajer, ``Opinion dynamics
  in the presence of increasing agreement pressure,'' \emph{IEEE Transactions
  on Cybernetics}, pp. 1--9, 2018.

\bibitem{Dong_etal_ieeebigdata_2018}
Y.~Dong, Z.~Ding, F.~Chiclana, and E.~Herrera-Viedma, ``Dynamics of public
  opinions in an online and offline social network,'' \emph{IEEE Transactions
  on Big Data}, pp. 1--11, 2018.

\bibitem{Altafini_tac_2013}
C.~Altafini, ``Consensus problem of networks with antagonistic interactions,''
  \emph{IEEE Trans. on Automatic Control}, vol.~58, no.~4, pp. 935--946, 2013.

\bibitem{Hendrickx_cdc_2014}
J.~M. Hendrickx, ``A lifting approach to models of opinion dynamics with
  antagonisms,'' in \emph{Proc. of the IEEE Conference on Decision and
  Control}.\hskip 1em plus 0.5em minus 0.4em\relax Los Angeles, California,
  USA, 2014, pp. 2118--2123.

\bibitem{Proskurnikov_etal_tac_2016}
A.~V. Proskurnikov, A.~S. Matveev, and M.~Cao, ``Opinion dynamics in social
  networks with hostile camps: Consensus vs. polarization,'' \emph{IEEE Trans.
  on Automatic Control}, vol.~61, no.~6, pp. 1524--1536, 2016.

\bibitem{Boyd_icm_2006}
S.~Boyd, ``Convex optimization of graph {Laplacian} eigenvalues,'' in
  \emph{Proc. of the International Congress of Mathematicians}, 2006, pp.
  3:1311–--1319.

\bibitem{Cao_acc_2015}
Y.~Cao, ``Consensus of multi-agent systems with state constraints: A unified
  view of opinion dynamics and containment control,'' in \emph{Proc. of the
  American Control Conference}.\hskip 1em plus 0.5em minus 0.4em\relax Chicago,
  IL, USA, 2015, pp. 1439--1444.

\bibitem{Liu_etal_tac_2018}
F.~Liu, D.~Xue, S.~Hirche, and M.~Buss, ``Polarizability, consensusability and
  neutralizability of opinion dynamics on coopetitive networks,'' \emph{IEEE
  Transactions on Automatic Control}, pp. 1--8, 2018.

\bibitem{Parsegov_etal_tac_2017}
S.~E. Parsegov, A.~V. Proskurnikov, R.~Tempo, and N.~E. Friedkin, ``Novel
  multidimensional models of opinion dynamics in social networks,'' \emph{IEEE
  Trans. Automatic Control}, vol.~62, no.~5, pp. 2270--2285, 2017.

\bibitem{Mengbin_etal_auto_sub_2017}
M.~Ye, M.~H. Trinh, Y.-H. Lim, B.~D.~O. Anderson, and H.-S. Ahn,
  ``Continuous-time opinion dynamics on multiple interdependent topics,''
  \emph{Submitted to Automatica (see arXiv:1805.02836 [cs.SI)}, pp. 1--16,
  2017.

\bibitem{Mengbin_etal_tac_2018}
M.~Ye, J.~Liu, B.~D.~O. Anderson, C.~Yu, and T.~Basar, ``Evolution of social
  power in social networks with dynamic topology,'' \emph{IEEE Trans. on
  Automatic Control (accepted)}, pp. 1--16, 2018.

\bibitem{Minh_etal_auto_2018}
M.~H. Trinh, C.~V. Nguyen, Y.-H. Lim, and H.-S. Ahn, ``Matrix-weighted
  consensus and its applications,'' \emph{Automatica}, vol.~89, no.~3, pp.
  415--419, 2018.

\bibitem{Minh_etal_ascc_2017}
M.~H. Trinh, M.~Ye, H.-S. Ahn, and B.~D.~O. Anderson, ``Matrix-weighted
  consensus with leader-following topologies,'' in \emph{Proc. of the Asian
  Control Conference}, Gold Coast Convention Centre, Australia, 2017.

\bibitem{Trinh2017arvix}
M.~H. Trinh and H.-S. Ahn, ``Theory and applications of matrix-weighted
  consensus,'' 2017, https://arxiv.org/abs/1703.00129.

\bibitem{horn1990matrix}
R.~A. Horn and C.~R. Johnson, \emph{Matrix Analysis}.\hskip 1em plus 0.5em
  minus 0.4em\relax Cambridge University Press, 1990.

\bibitem{Ahn_etal2018_opinion_arvix}
H.-S. Ahn, Q.~V. Tran, and M.~H. Trinh, ``Cooperative opinion dynamics on
  multiple interdependent topics: {Modeling} and some observations,'' 2018,
  arXiv:1807.04406 [cs.SY].

\end{thebibliography}

\end{document}